\documentclass{article}

\usepackage{arxiv}

\usepackage[T1]{fontenc}
\usepackage[utf8]{inputenc}

\usepackage{amsmath,amsthm,mathtools}

\usepackage{newtxtext,newtxmath}
\usepackage{microtype}

\usepackage{hyperref}       
\usepackage{url}            
\usepackage{booktabs}       
\usepackage{amsfonts}       

\usepackage{hyperref}

\usepackage{ragged2e}
\numberwithin{equation}{section}

\newtheorem{proposition}{Proposition}
\newtheorem{remark}{Remark}

\usepackage{graphicx}
\graphicspath{{figs/}}
\usepackage{tikz}
\usetikzlibrary{calc,arrows.meta,positioning}
\usepackage{subcaption}
\usepackage{float}

\usepackage{enumitem}
\setlist{nosep}
\usepackage{booktabs,multirow}

\title{PAS-Net: \underline{P}hysics-informed \underline{A}daptive \underline{S}cale Deep Operator \underline{Net}work}
\author{
  Changhong Mou \\
  Department of Mathematics and Statistics\\ Utah State University \\
  900 Old Main Hill, Logan, UT 84322, USA \\
  \texttt{changhong.mou@usu.edu} 
  \And
  Yeyu Zhang \\
  School of Mathematics\\ 
  Shanghai University of Finance and Economics\\ No.777 Guoding Road, Shanghai 200433, China\\
  \texttt{zhangyeyu@mail.shufe.edu.cn} 
\And
  Xuewen Zhu \\
  School of Mathematics\\ 
  Shanghai University of Finance and Economics\\ No.777 Guoding Road, Shanghai 200433, China\\
  \texttt{zhuxuewen@stu.sufe.edu.cn} 
  \And
 Qiao Zhuang\\
 School of Science and Engineering \\ 
 University of Missouri-Kansas City\\
 5110 Rockhill Rd, Kansas City, MO 64110, USA \\
   \texttt{qzhuang@umkc.edu} 
 }

\begin{document}

\maketitle

\begin{abstract}
Nonlinear physical phenomena often show complex multiscale interactions;  motivated by the principles of multiscale modeling in scientific computing,  we propose \emph{PAS-Net}, a physics-informed Adaptive-Scale Deep Operator Network for learning solution operators of nonlinear and singularly perturbed evolution PDEs with small parameters and localized features. Specifically, PAS-Net augments the trunk input in the physics informed Deep Operator Network (PI-DeepONet) with a  prescribed  (or learnable) locally rescaled coordinate transformation centered at reference points. This addition introduces a multiscale feature embedding that acts as an architecture-independent preconditioner  which improves the representation of localized, stiff, and multiscale dynamics.
From an optimization perspective, the adaptive-scale embedding in PAS-Net modifies the geometry of the Neural Tangent Kernel (NTK) associated with the neural network by increasing its smallest eigenvalue, which in turn improves spectral conditioning and accelerates gradient-based convergence.
We further show that this adaptive-scale mechanism explicitly accelerates neural network training in approximating functions with steep transitions and strong asymptotic behavior, and we provide a rigorous proof of this function-approximation result within the finite-dimensional NTK matrix framework.
We test the proposed PAS-Net on three different problems: (i) the one-dimensional viscous Burgers equation, (ii) a nonlinear diffusion–reaction system with sharp spatial gradients, and (iii) a two-dimensional eikonal equation. The numerical results show that PAS-Net consistently achieves higher accuracy and faster convergence than the standard DeepONet and PI-DeepONet models under a similar training cost.
\end{abstract}
\section{Introduction}


The rapid development of \emph{scientific machine learning} (SciML) has opened new directions for data-driven modeling, analysis, and prediction of complex physical systems governed by partial differential equations (PDEs) \cite{cuomo2022scientific,lennon2023scientific,dou2023machine,sanderse2024scientific,chen2024learning}. 
In recent years, two major directions have become popular within SciML for solving PDEs: \emph{Physics-Informed Neural Networks} \cite{raissi2019physics,cai2021physics,lu2025mopinnenkf}(PINNs) and \emph{operator learning} \cite{lu2021learning,li2021fourier,wang2021learning}. 
The first category incorporates the governing PDE constraints directly into the loss function by penalizing the residuals of the differential operators at a set of collocation points. 
This residual-based formulation, often referred to as the \emph{PDE residual loss}, enforces the underlying physical laws during training and enables the network to infer continuous solutions even from sparse or noisy observations. 
Such a physics-informed regularization is particularly advantageous in \emph{inverse problems}, where unknown physical parameters or boundary conditions are estimated simultaneously with the PDE solution, without the need for large labeled datasets.
In contrast, operator learning aims to approximate mappings between infinite-dimensional function spaces, enabling fast surrogate modeling and real-time simulation of parameterized PDEs. 
Representative frameworks include the Deep Operator Network (DeepONet)~\cite{lu2021learning} and the Fourier Neural Operator (FNO)~\cite{li2021fourier}. 
DeepONet decomposes a nonlinear operator into two subnetworks: a \emph{branch network} that encodes the input function via its pointwise sensor evaluations, and a \emph{trunk network} that represents the coordinates of the output field. 
The outputs of these subnetworks are combined through an inner product to approximate the target solution operator, achieving universal approximation in Banach spaces~\cite{chen1995universal}. 
The FNO, on the other hand, works in the frequency domain by applying spectral convolution in Fourier space, which allows efficient capture of long-range dependencies and fast inference on unseen spatial resolutions through the inverse fast Fourier transform (IFFT).  
Building upon these two approaches, the physics-Informed Deep Operator Network (PI-DeepONet)~\cite{wang2021learning} incorporate PDE residual losses into the operator-learning framework, which enables simultaneous data-driven and physics-constrained learning of nonlinear PDE solution operators.

Despite these advances, operator learning still faces fundamental challenges in handling \emph{singularly perturbed} and \emph{multiscale} PDEs, which involve strong scale interactions and stiffness that may hinder accurate and efficient learning.
Such equations often involve small parameters that induce multiple spatial or temporal scales, sharp gradients, or boundary and internal layers~\cite{goering1984singularly,kadalbajoo2010brief}. 
Traditional numerical methods, such as finite element or finite difference schemes, must employ adaptive meshes,  solution space enrichment, or homogenization strategies to resolve these localized structures~\cite{efendiev2009multiscale,hughes1998variational}. 
Analogously, neural operator models trained on uniformly sampled data may suffer from loss of resolution or biased approximation in regions where the solution varies rapidly \cite{bartolucci2023representation,roy2025physics}. 
Another important challenge lies in the \emph{optimization and convergence} of neural networks when learning solutions of PDEs with limited regularity or strongly varying scales. 
For problems that exhibit steep transitions, sharp layers, or stiff dynamics, standard neural architectures often suffer from slow or unstable convergence \cite{urban2025unveiling,gu2021selectnet}, largely due to ill-conditioned loss landscapes and unfavorable spectral properties of the Neural Tangent Kernel (NTK)~\cite{jacot2018ntk,lee2019wide}. 
These difficulties highlight the need for architectural mechanisms that can adapt to local scales and geometric features in the solution space, thereby improving both representational efficiency and optimization stability.

Inspired by multiscale modeling principles in scientific computing \cite{jin1999efficient,hou1997multiscale,weinan2011principles,hughes1998variational}, 
we propose the \emph{Physics-Informed Adaptive-Scale Deep Operator Network} (PAS-Net) which is a new framework for learning nonlinear operator mappings in partial differential equations (PDEs) that exhibit small-scale parameters, localized structures, or singular perturbations. Our approach also aligns with with the recent progress in multiscale deep neural network (DNN) i.e., including scaling in input features or hidden layers \cite{zhuang2024two,LiXuZhang20,LiXuZhang23}. 
Unlike conventional neural operator architectures, PAS-Net embeds an adaptive scale feature directly into the trunk network of the Physics-Informed Deep Operator Network (PI-DeepONet). 
Specifically, the adaptive-scale feature augments the input coordinates with locally rescaled transformations centered around reference points, where both the scaling factors and reference locations can be either prescribed or learned. 
This embedding enables the network to resolve localized structures across varying spatial scales. 
Consequently, the adaptive-scale formulation provides an intrinsic multiscale feature representation that acts as a model-independent preconditioner, enhancing both expressivity and training stability when addressing stiff or highly localized dynamics. 
On the theoretical side, we analyze how the adaptive-scale embedding influences the Neural Tangent Kernel (NTK) \cite{jacot2018neural,bietti2019inductive} and show that it increases the smallest eigenvalue, thereby improving spectral conditioning and accelerating convergence.  
This theoretical insight is consistent with both empirical and analytical evidence, including function-approximation experiments and a rigorous proof within the finite-dimensional NTK matrix framework that adaptive scaling leads to faster and more stable training for functions with steep gradients and strong asymptotic behavior.
The main contributions of this work are summarized as follows:
\begin{itemize}[leftmargin=2em]
  \item We introduce a new \emph{adaptive-scale embedding} within the Physics-Informed Deep Operator Network (PI-DeepONet) framework which enables localized rescaling of spatial coordinates around reference points. This design provides an intrinsic multiscale representation that enhances the network’s ability to capture sharp gradients, boundary layers, and stiff dynamics in singularly perturbed PDEs. 

  \item  We establish a theoretical link between adaptive-scale embeddings and the Neural Tangent Kernel (NTK), showing that the proposed mechanism increases the smallest eigenvalue of the NTK, thereby improving spectral conditioning and accelerating gradient-based convergence. 

  \item  Through extensive experiments on benchmark test problems including the Burgers, advection, and diffusion–reaction equations, we show that PAS-Net consistently achieves higher accuracy and faster convergence than standard DeepONet and PI-DeepONet models, particularly for problems with steep transitions and multiscale features.
\end{itemize}

The remainder of this paper is organized as follows. 
Section~\ref{sec:pas-net} introduces the proposed PAS-Net framework includes the physics-informed formulation, adaptive-scale embedding, and theoretical analysis based on the Neural Tangent Kernel (NTK). 
Section~\ref{sec:numeric} presents numerical experiments on Burgers, diffusion–reaction equations, and 2D eikonal equation.
Finally, Section~\ref{sec:conclusion} summarizes the main findings and discusses potential directions for future research.


\section{PAS-Net \label{sec:pas-net}}
\subsection{Partial differential equations with small parameters}
Let $D \subset \mathbb{R}^d$ ($d \in \{1,2,3\}$) denote a bounded Lipschitz domain with boundary $\partial D$ and outward unit normal vector $\mathbf{n}$. We consider the following nonlinear, possibly singularly perturbed evolution partial differential equation (PDE) \cite{goering1984singularly,kadalbajoo2010brief,evans2022partial}:
\begin{equation}\label{eq:equation-of-interest}
\partial_t u + \mathcal{L}_\varepsilon(u) + \mathcal{N}(u) = f,
\qquad (\mathbf{x},t) \in D \times (0,T],
\end{equation}
subject to the initial condition
\begin{align}
    u(\mathbf{x},0) = u_0(\mathbf{x}), \qquad \mathbf{x} \in D,
\end{align}
and with appropriate boundary conditions. 
Here $\varepsilon \in (0,1]$ is a (possibly small) dimensionless parameter that may represent diffusion strength, regularization, or scale separation. The function $u : D \times [0,T] \to \mathbb{R}$ is the unknown state variable (e.g., velocity potential, concentration, or temperature), and $f : D \times (0,T] \to \mathbb{R}$ denotes a prescribed source or forcing term.
The operator $\mathcal{L}_\varepsilon$ is a linear differential operator that may depend on the small parameter $\varepsilon$, typically of elliptic type, while $\mathcal{N}$ is a nonlinear operator accounting for advection, reaction, or other nonlinear effects. Equation \eqref{eq:equation-of-interest} therefore represents a general class of nonlinear PDEs.
\subsection{Physics-informed Deep Operator Network}
The general deep operator network (DeepONet) proposed \cite{lu2021learning} with mathematical foundation of the universal operator approximation theorem \cite{chen1995universal} is intended to learn nonlinear operator mappings between infinite-dimensional function spaces.
For a general PDE problem, we introduce the solution operator  
\begin{align}
\mathcal{G}(\kappa) = u(\cdot;\kappa),    
\end{align}
where $\kappa$ represents the input which can the initial condition or parameters of the PDE. A DeepONet is usually a class of neural network models designed to approximate $\mathcal{G}$.
In particular, we denote by $G_\theta$ the DeepONet approximation of $\mathcal{G}$, with $\theta$ representing the trainable parameters of the network. Then a DeepONet usually consists of two subnetworks: the branch net and the trunk net and the estimator evaluated at $\mathbf{x} \in M$ is an inner product of the branch and trunk outputs:
\begin{align}
G_\theta(\kappa(\Xi))(\mathbf{x}) = \sum_{k=1}^p b_k(\kappa(\xi_1), \kappa(\xi_2), \ldots, \kappa(\xi_m)) \, t_k(\mathbf{x}) + b_0,    
\end{align}
where $b_0 \in \mathbb{R}$ is a bias, $\{b_1, b_2, \ldots, b_p\}$ are the $p$ outputs of the branch net, and $\{t_1, t_2, \ldots, t_p\}$ are the $p$ outputs of the trunk net. 
More specifically, the branch network represents $\kappa$ in a discrete format: 
\begin{align}
\kappa(\Xi) = \{\kappa(\xi_1), \kappa(\xi_2), \ldots, \kappa(\xi_m)\}, 
\end{align}
where $\kappa(\Xi)$ is a vector which consist the evaluation of input functions at different sensors $\Xi = \{\xi_1, \xi_2, \ldots, \xi_m\}$.
The trunk network  may take as input a spatial location $\mathbf{x} \in M$. 
Our goal is to predict the solution corresponding to a given $\kappa$, which is also evaluated at the same input location $\mathbf{x} \in M$. 
The representation of $\kappa$ through pointwise evaluations at arbitrary sensor locations provides additional flexibility, both during training and in prediction, particularly when $\kappa$ is available only through its values at these sensor points.

Given the training data with labels in the two sets of inputs and outputs:
\begin{align}
&  \text{input:}  &\Big\{\kappa^{(k)}, \Xi, (\mathbf{x}_i^{(k)})_{i=1,\ldots,N}\Big\}_{k=1,\ldots,N_{\text{data}}}
    \\
 &   \text{output:}
&\{u(\mathbf{x}_i^{(k)}; \kappa^{(k)})\}_{i=1,\ldots,N,\; k=1,\ldots,N_{\text{data}}}.
\end{align}
we minimize a loss function that measures the discrepancy between the true solution operator $G(\kappa^{(k)})$ and its DeepONet approximation $G_\theta(\kappa^{(k)})$:  
\begin{align}
\widehat{\mathcal{L}}_{\text{data}}(\theta) 
= \frac{1}{N_{\text{data}}N} \sum_{k=1}^{N_{\text{data}}} \sum_{i=1}^N 
\Big| G_\theta(\kappa^{(k)}(\Xi))(\mathbf{x}_i^{(k)}) 
     - G(\kappa^{(k)})(\mathbf{x}_i^{(k)}) \Big|^2,
\end{align}
where $X^{(k)} = \{\mathbf{x}_1^{(k)}, \ldots, \mathbf{x}_N^{(k)}\}$ denotes the set of $N$ evaluation points in the domain of $G(\kappa^{(k)})$.  
In practice, however, the analytical solution of the PDE is not available. Instead, we rely on numerical simulation obtained from high fidelity numerical solvers which can be denoted as: 
\begin{align}
\big\{\hat{u}(\mathbf{x}_i^{(k)}; \kappa^{(k)}) \big\}_{i=1,\ldots,N,\; k=1,\ldots,N_{\text{data}}}.    
\end{align}
Then the practical training tends to minimize the loss function  
\begin{align}
{\mathcal{L}}_{\text{data}}(\theta) 
= \frac{1}{N_{\text{data}}N} \sum_{k=1}^{N_{\text{data}}} \sum_{i=1}^N 
\Big| G_\theta(\kappa^{(k)}(\Xi))(\mathbf{x}_i^{(k)}) 
     - \hat{u}(\mathbf{x}_i^{(k)}; \kappa^{(k)}) \Big|^2.
\end{align}

The Physics-Informed Deep Operator Network (PI-DeepONet) \cite{wang2021learning} integrates the idea of physics-informed neural networks (PINN) \cite{raissi2019physics} with the DeepONet framework. 
The PI-DeepONet introduces the additional PDE residual loss function and boundary condition loss function.
Consequently, the loss function for PI-DeepONet is defined as
\begin{align}
\mathcal{L}(\theta) 
= w_{\text{data}} \mathcal{L}_{\text{data}}(\theta) 
+ w_{\text{PDE}} \mathcal{L}_{\text{PDE}}(\theta) 
+ w_{\text{bc}} \mathcal{L}_{\text{bc}}(\theta),
\label{loss-pi-don}
\end{align}
where $\mathcal{L}_{\text{PDE}}(\theta)$ denotes the PDE residual loss, $\mathcal{L}_{\text{bc}}(\theta)$ the boundary condition loss, and $w_{\text{data}}, w_{\text{PDE}},$ and $w_{\text{bc}}$ are the corresponding weights for each term.

\subsection{Adaptive Scale Embedding\label{ss:scale}}

To improve the representational ability of neural networks for learning functions or operators with localized features, we introduce an \emph{adaptive scale feature} that incorporates local spatial rescaling around a reference point~$\mathbf{x}_c$~\cite{zhuang2024two}. 
Rather than relying solely on the original physical coordinate $\mathbf{x}$ as network input, the proposed approach embeds an additional scale-dependent feature that encodes multiscale spatial variations. 
Specifically, for a general neural network with input variable $\mathbf{x}$, the additional multiscale feature augments the original coordinates with locally rescaled displacements around a reference point~$\mathbf{x}_c$ such that:
\begin{align}
  \mathbf{x}
  \;\longrightarrow\;  \mathbf{x}\oplus\xi \quad \text{   with   } \quad\xi\coloneqq 
    \bigoplus_{i = 1}^{\Gamma}\!\bigl(
        \epsilon^{\gamma_i}(\mathbf{x}-\mathbf{x}_c)
    \bigr),
\end{align}
where $\epsilon > 0$ is a prescribed or learnable scaling parameter that controls the degree of local stretching around $\mathbf{x}_c$, $\gamma_i$ is a parameter in exponent that adjusts the effective scaling, and $\Gamma \in \mathbb{N}$ denotes the number of multiscale components included in the embedding. 
Accordingly, the augmented input feature map is defined as
\begin{align}
    \Psi_{\Gamma}\!\left(\mathbf{x};\,\mathbf{x}_c,\epsilon\right)
    =\mathbf{x}\oplus\xi
\end{align}
where the concatenation of the original coordinate $\mathbf{x}$ with its scaled displacements 
$(\mathbf{x}-\mathbf{x}_c)$ forms a hierarchical spatial representation capable of capturing different scale variations in the target function or operator. 
In this work, we employ the simplest single-scale version,
\begin{align}
    \Psi\!\left(\mathbf{x};\,\mathbf{x}_c,\epsilon\right)
    =
    \mathbf{x}\oplus\xi
    =
        \mathbf{x}\oplus
        \epsilon^{\gamma}(\mathbf{x}-\mathbf{x}_c),
\end{align}
which already yields a faster convergence in function approximations and more accurate learning in PI-DeepOnet framework.
The adaptive-scale embedding changes the structure of the underlying Neural Tangent Kernel (NTK), leading to an increase in its smallest eigenvalue and improved spectral conditioning of the associated learning dynamics, as discussed in the subsequent Section \ref{ss:ntk-convergence}.

\subsubsection{Convergence of Function Approximation with Adaptive Scale \label{ss:ntk-convergence}}

The Neural Tangent Kernel (NTK) framework provides a theoretical framework to interprete the optimization behavior of wide neural networks.
Introduced by Jacot, Gabriel, and Hongler \cite{jacot2018ntk}, it states that when a network is sufficiently wide, its parameters evolve almost linearly during training which is known as the lazy-training regime, such  that gradient descent becomes equivalent to kernel gradient descent governed by a deterministic kernel, called the NTK.
Subsequent studies have extended this framework to deeper and finite-width networks, clarifying convergence rates, spectral properties, and the transition toward feature learning \cite{lee2019wide,arora2019finegrained,chizat2019lazy}.
Within this setting, the smallest nonzero eigenvalue of the NTK operator determines the exponential rate at which the training error decays.
This spectral viewpoint directly motivates our analysis: the adaptive input scaling effectively reshapes the NTK, which explains the observed faster convergence of networks trained with such adaptive features.
The NTK-based Convergence for Function Approximation states as follows:
\begin{proposition}[NTK-based Convergence for Function Approximation]
\label{prop:ntk-convergence}
Let $h_\theta:\mathcal{X}\to\mathbb{R}$ be a sufficiently wide neural network trained on a target function $f$ by gradient descent on $\mathcal{L}(\theta)=\tfrac{1}{2}\,\mathbb{E}_{x\sim\mu}[(h_\theta(x)-f(x))^2]$.  
In the lazy-training regime, i.e. the linear regime, linearization around $\theta_0$ gives $h_\theta(x)\approx h_{\theta_0}(x)+\nabla_\theta h_{\theta_0}(x)^\top(\theta-\theta_0)$ with Neural Tangent Kernel (NTK) $K(x,x')=\nabla_\theta h_{\theta_0}(x)^\top\nabla_\theta h_{\theta_0}(x')$ and associated integral operator $(\mathcal{T}_K\varphi)(x)=\int K(x,x')\varphi(x')\,d\mu(x')$.  
Under gradient flow $\tfrac{d\theta}{dt}=-\eta\nabla_\theta\mathcal{L}(\theta)$, the training error $e_t(x)=h_{\theta(t)}(x)-f(x)$ evolves as
\begin{align}
\frac{d}{dt}e_t(x)&=-\eta(\mathcal{T}_K e_t)(x), &
e_t&=e^{-\eta t\mathcal{T}_K}e_0,
\end{align}
and satisfies the exponential decay
\begin{equation}
\|e_t\|_{L^2(\mu)}^2
\le e^{-2\eta t\lambda_{\min}}\|e_0\|_{L^2(\mu)}^2,
\end{equation}
where $\lambda_{\min}>0$ is the smallest nonzero eigenvalue of $\mathcal{T}_K$.  
Thus, the convergence rate of the neural network toward $f$ is governed by the spectral gap of $\mathcal{T}_K$.
\end{proposition}

\begin{remark}[Infinite-Width Limit]
In the infinite-width regime, the parameter perturbation during training satisfies 
$\|\theta(t)-\theta_0\|=\mathcal{O}(1/\sqrt{m})$, implying that the Jacobian 
$\nabla_\theta h_\theta(x)$ remains effectively constant. 
Consequently, the linearization 
$h_\theta(x)=h_{\theta_0}(x)+\nabla_\theta h_{\theta_0}(x)^\top(\theta-\theta_0)$ 
becomes exact, and the Neural Tangent Kernel $K(x,x')$ is time-invariant. 
Under the gradient-flow dynamics 
\begin{equation}
\frac{d\theta}{dt}=-\eta\nabla_\theta\mathcal{L}(\theta),
\end{equation}
the network training error evolves as follows:
\begin{equation}
\frac{d}{dt}e_t(x)=-\eta\!\int\!K(x,x')\,e_t(x')\,d\mu(x')
=-\eta(\mathcal{T}_K e_t)(x),
\end{equation}
which further yields the closed-form solution 
$e_t=e^{-\eta t\mathcal{T}_K}e_0$ with the exact exponential decay 
\begin{equation}
\|e_t\|_{L^2(\mu)}^2
\le e^{-2\eta t\lambda_{\min}}\|e_0\|_{L^2(\mu)}^2.
\end{equation}
Thus, in the infinite-width limit, the training dynamics of a neural network 
coincide exactly with kernel regression governed by the NTK operator.
\end{remark}
The detailed proof can be found in \cite{jacot2018ntk}.
Building on this result, we show that introducing an adaptive-scale feature increases the smallest eigenvalue of the NTK, therefore improving the convergence rate of training.

\begin{proposition}[NTK-based Convergence for Function Approximation with Adaptive Scale Feature]
\label{prop:adaptive-eig}
Let $K(x,x')$ denote the baseline Neural Tangent Kernel (NTK) associated with the feature map $x\mapsto[x]$, and let $\mathcal{T}_K$ be its integral operator on $L^2(\mu)$ with smallest nonzero eigenvalue $\lambda_{\min}(K)>0$. 
For a fixed center $x_c\in\mathbb{R}^d$, define the adaptive scale feature map
\begin{align}
\phi_\varepsilon(x)=[\,x,\ \varepsilon^{-1}(x-x_c)\,], 
\qquad 
K_\varepsilon(x,x') := K\!\bigl(\phi_\varepsilon(x),\phi_\varepsilon(x')\bigr),    
\end{align}
and let $\mathcal{T}_{K_\varepsilon}$ be the corresponding NTK operator. 
Assume $K$ is smooth, positive definite, and radially nonincreasing in $\|x-x'\|$. 
Then, for sufficiently small $\varepsilon>0$, there exist constants $c_1,c_2>0$ (independent of $\varepsilon$) such that
\begin{equation}
\label{eq:lambda-comparison}
\lambda_{\min}(K_\varepsilon)
\;\ge\;
\lambda_{\min}(K)\;,
\end{equation}

Consequently, the adaptive-scaled NTK exhibits a larger spectral gap and faster error decay in function space:
\begin{equation}
\label{eq:adaptive-ntk-decay}
\|e_t\|_{L^2(\mu)}^2
\;\le\;
\exp\!\bigl(-2\eta\,\lambda_{\min}(K_\varepsilon)\,t\bigr)\,
\|e_0\|_{L^2(\mu)}^2,
\qquad
\lambda_{\min}(K_\varepsilon)>\lambda_{\min}(K).
\end{equation}
Hence, shrinking the input scale via $\varepsilon^{-1}$ amplifies local variations near $x_c$, 
shortens the effective kernel correlation length, 
and increases the smallest eigenvalue of the NTK operator that leads to faster exponential convergence under the dynamics of Proposition~\ref{prop:ntk-convergence}.
\end{proposition}

\begin{proposition}[Matrix monotonicity of adaptive scale NTK ]
\label{prop:matrix-ntk-monotone}
Fix data points $x_1,\dots,x_N$ and let $G_x\in\mathbb{R}^{N\times p}$ be the Jacobian of the network output with respect to a parameter block that only relates to the input $x$, so that the empirical NTK matrix is
\begin{equation}
K \;=\; G_x G_x^\top \;\in\; \mathbb{R}^{N\times N}, \qquad K\succeq 0.
\end{equation}
Introduce an adaptive scale feature $\xi_\varepsilon(x)=(x-x_c)/\varepsilon$ and an \emph{additional, disjoint} parameter block that only connects to this new channel, with Jacobian $G_\xi\in\mathbb{R}^{N\times q}$. The adaptive network’s full Jacobian is the column concatenation
\begin{equation}
\begin{aligned}
G_\varepsilon &= [\,G_x \;\; G_\xi\,], \\[3pt]
K_\varepsilon &= G_\varepsilon G_\varepsilon^\top 
= G_x G_x^\top + G_x G_\xi^\top + G_\xi G_x^\top + G_\xi G_\xi^\top. \\[3pt]
&\;\;= K + G_\xi G_\xi^\top, \quad \text{if } G_x G_\xi^\top = 0.
\end{aligned}
\end{equation}
Then
\begin{equation}
\lambda_{\min}(K_\varepsilon)\;\ge\;\lambda_{\min}(K).
\end{equation}
Moreover, if $\mathrm{range}(G_\xi)\not\subseteq \mathrm{range}(G_x)$, then $\lambda_{\min}(K_\varepsilon)>\lambda_{\min}(K)$.
\end{proposition}

\begin{proof}
Set $\Delta\!:=G_\xi G_\xi^\top\succeq 0$. For any $v\in\mathbb{R}^N$,
\begin{equation}
v^\top K_\varepsilon v \;=\; v^\top(K+\Delta)v \;=\; v^\top K v \;+\; \|G_\xi^\top v\|_2^2 \;\ge\; v^\top K v .
\end{equation}
Taking the minimum over $\|v\|_2=1$ (Courant--Fischer) yields
\begin{equation}
\lambda_{\min}(K_\varepsilon)
\;=\;\min_{\|v\|=1} v^\top K_\varepsilon v
\;\ge\;\min_{\|v\|=1} v^\top K v
\;=\;\lambda_{\min}(K).
\end{equation}
If $\mathrm{range}(G_\xi)\not\subseteq \mathrm{range}(G_x)$, then for every unit minimizer $v_\star$ of $v^\top K v$ we have $G_\xi^\top v_\star\neq 0$, hence $v_\star^\top K_\varepsilon v_\star > v_\star^\top K v_\star=\lambda_{\min}(K)$, which implies $\lambda_{\min}(K_\varepsilon)>\lambda_{\min}(K)$.
\end{proof}

\subsubsection{An Example of Function Approximation with Adaptive Scale Feature}
To illustrate the preceding proposition, we consider a one-dimensional function \( f(x) \) that shows a steep transition and strong asymptotic behavior near \( x = 1 \). 
It is defined as
\begin{align}
f(x) = \frac{1}{1 + \exp\!\left(-\frac{x - x_c}{0.2\epsilon}\right)}
+ 0.2 \exp\!\left(-\frac{(x - x_c)^2}{2(0.3\epsilon)^2}\right),
\end{align}
where the parameters \(x_c = 1.0\) and \(\epsilon = 0.01\) determine the location and sharpness of the transition region. 
Training samples are uniformly drawn from the interval \([0,2]\) with \(N_{\text{train}} = 200\).
Two multilayer perceptrons (MLPs) of identical depth and width are compared to assess the effect of adaptive input scaling. 
The {plain network} uses the original coordinate \(x\) as input, corresponding to an architecture of \([1,\,20,\,20,\,1]\). 
The \textit{adaptive scale network} augments the input with a locally rescaled feature \(\xi = \epsilon^{-1}(x - x_c)\), resulting in an input dimension of two, i.e., \([2,\,20,\,20,\,1]\). 
Both models uses \(\tanh\) activation functions and are trained in single precision using the Adam optimizer with a learning rate of \(10^{-3}\). 
Training is performed in full-batch mode for \(4000\) gradient descent steps by minimizing the mean squared error (MSE) between network predictions and the target values. The trained models are then evaluated on a fine grid of 800 points in $[0,2]$ to compare their approximation accuracy and to examine how the adaptive-scale feature enhances the learning of sharp local variations.


Figure~\ref{fig:adaptive-scale} compares the performance of the plain network and adaptive-scale network during training and in their final function approximations. 
As shown in Figure~\ref{fig:fun-train}, the adaptive-scale model exhibits a significantly faster decrease in the training loss, reaching a lower error level within the same number of optimization steps. 
This faster convergence confirms the theoretical prediction that adaptive input scaling enlarges the smallest eigenvalue of NTK, which further improve the convergence the underlying optimization dynamics. 
Figure~\ref{fig:fun-scale} further shows that the adaptive-scale network more accurately resolves the steep transition near \(x=1\), whereas the plain network exhibits oscillations near the steep gradient, as shown in the zoomed-in view of Figure~\ref{fig:fun-scale}.
\begin{figure}[htp!]
    \centering
    \begin{subfigure}[b]{0.48\linewidth}
        \centering
        \includegraphics[width=\linewidth]{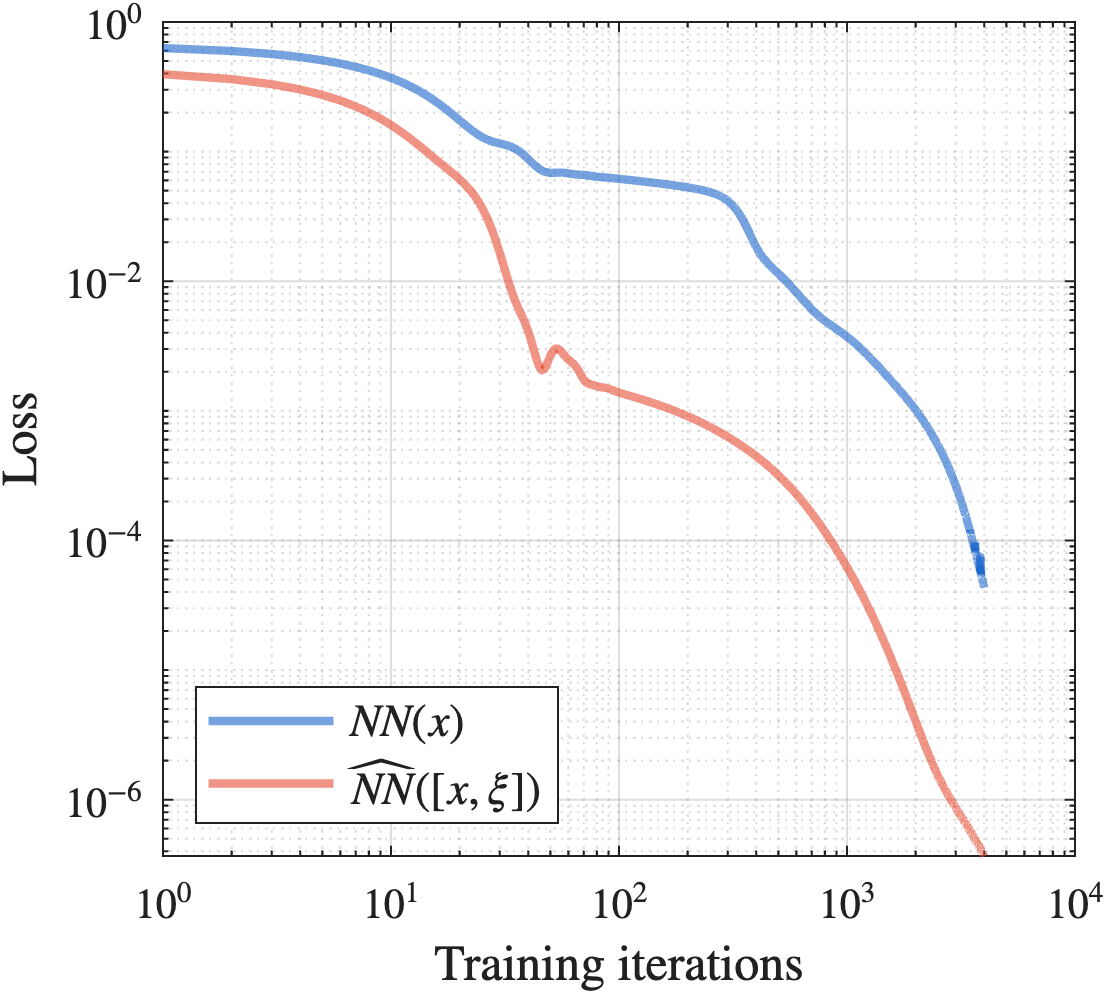}
        \caption{Training loss with and without adaptive scale.}
        \label{fig:fun-train}
    \end{subfigure}
    \hfill
    \begin{subfigure}[b]{0.48\linewidth}
        \centering
        \includegraphics[width=\linewidth]{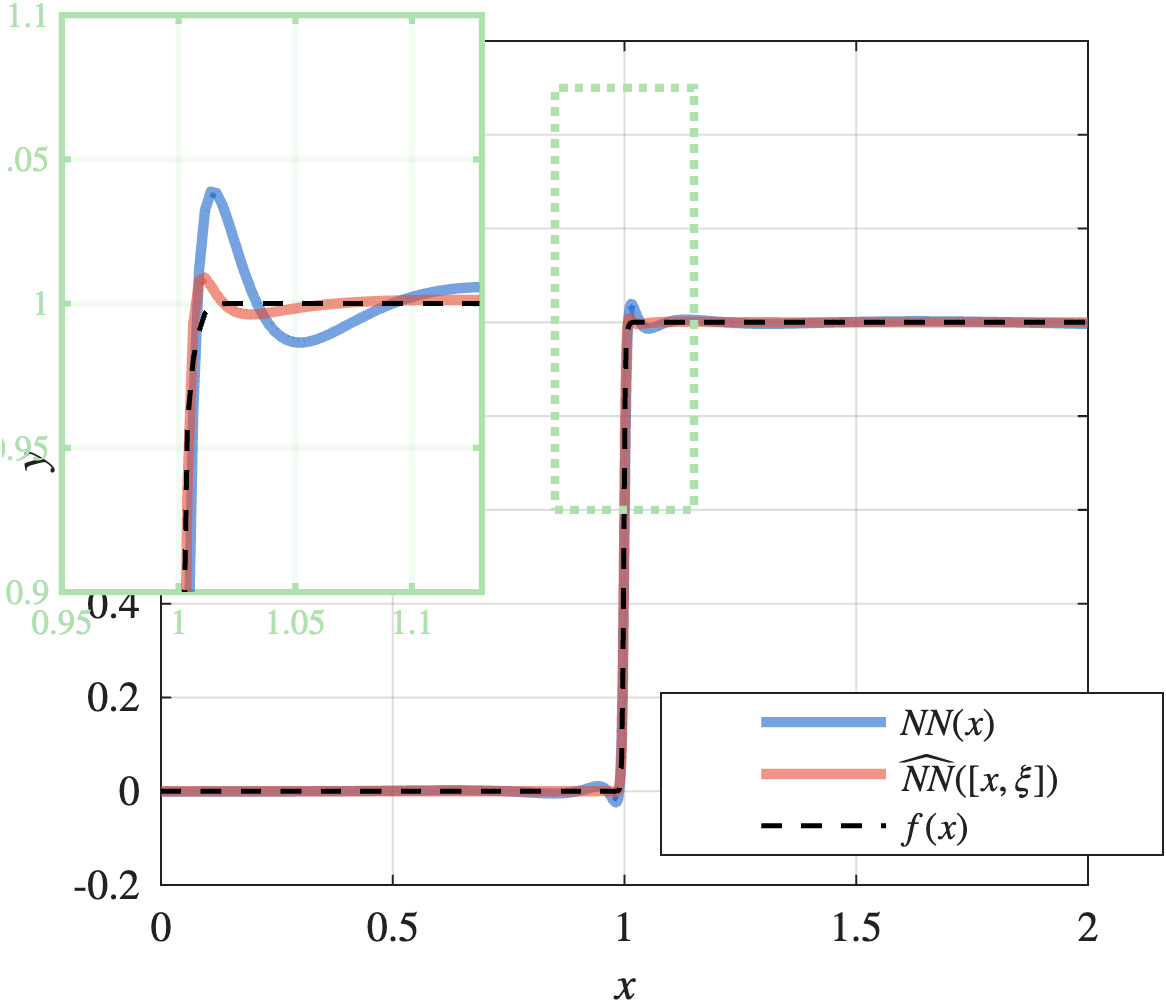}
        \caption{Function approximation with and without adaptive scale.}
        \label{fig:fun-scale}
    \end{subfigure}
    \caption{Illustrations of adaptive-scale features in approximating a steep-gradient function.}
    \label{fig:adaptive-scale}
\end{figure}

\subsection{PAS-Net Framework}
With the adaptive scale embedding introduced in Section~\ref{ss:scale}, the trunk network thus becomes
\begin{align}
    t_k(\mathbf{x}; \mathbf{x}_c, \epsilon(\mathbf{x}_c))
    = \tilde{t}_k\!\Bigl(
        \Psi_{\Gamma}\bigl(\mathbf{x}; \mathbf{x}_c, \epsilon\bigr)
    \Bigr).
\end{align}
Therefore, the physics-informed adaptive scale DeepONet (PAS-Net) yields the following:
\begin{align}
    G_\theta(\kappa(\Xi))(\mathbf{x})
    = \sum_{k=1}^{p} 
        b_k\!\bigl(\kappa(\Xi)\bigr)\,
        \tilde{t}_k\!\Bigl(
        \Psi_{\Gamma}\bigl(\mathbf{x}; \mathbf{x}_c, \epsilon\bigr)\Bigr)
      + b_0,
\end{align}
where the branch network encodes the input function $\kappa$ evaluated at sensor locations $\Xi = \{\xi_1, \xi_2, \ldots, \xi_m\}$, and the multiscale trunk includes local geometric information relative to $\mathbf{x}_c$.
The term $\sum_{\gamma=1}^{\Gamma}\epsilon^{\gamma}(\mathbf{x}-\mathbf{x}_c)$ is an adaptive scale  preconditioner for the coordinate $\mathbf{x}$ that can potentially capture the local features such as sharp gradients or discontinuities near $\mathbf{x}_c$. 
Similar to PI-DeepONet, the loss function for PAS-Net is defined as
\begin{align}
\mathcal{L}(\theta) 
= w_{\text{data}} \mathcal{L}_{\text{data}}(\theta) 
+ w_{\text{PDE}} \mathcal{L}_{\text{PDE}}(\theta) 
+ w_{\text{bc}} \mathcal{L}_{\text{bc}}(\theta),
\label{loss-pas-net}
\end{align}
where $\mathcal{L}_{\text{PDE}}(\theta)$ denotes the PDE residual loss, $\mathcal{L}_{\text{bc}}(\theta)$ the boundary condition loss, and $w_{\text{data}}, w_{\text{PDE}},$ and $w_{\text{bc}}$ are the corresponding weights for each term. 
The overview of PAS-Net framework is illustrated in Figure \ref{fig:pas-net}.

\begin{figure}[H]
    \centering
    \includegraphics[width=\linewidth]{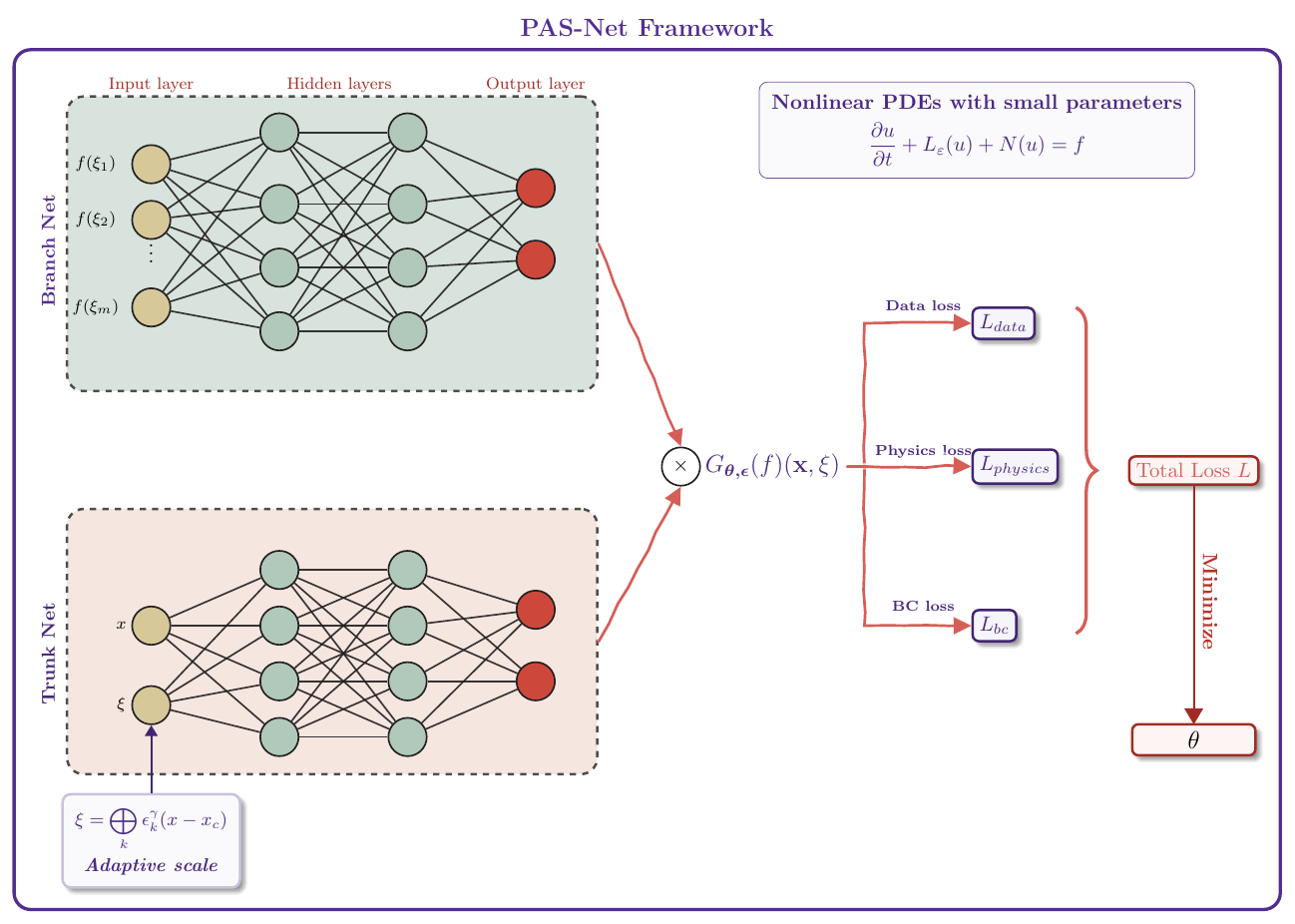}
    \caption{Illustration of PAS-Net framework.}
    \label{fig:pas-net}
\end{figure}

\section{Numerical Results \label{sec:numeric}}
In this section, we evaluate the performance of the proposed PAS-Net on three benchmark problems: the Burgers equation, a nonlinear diffusion–reaction system with forcing, and the Eikonal equation.These test cases, illustrated in Fig.~\ref{fig:test-problem}, are nonlinear and singularly perturbed evolution PDEs characterized by small parameters that show sharp gradients, localized structures, and multi-scale spatiotemporal features.
\begin{figure}[H]
    \centering
    \begin{subfigure}[b]{.32\linewidth}
        \centering
        \includegraphics[width=\linewidth]{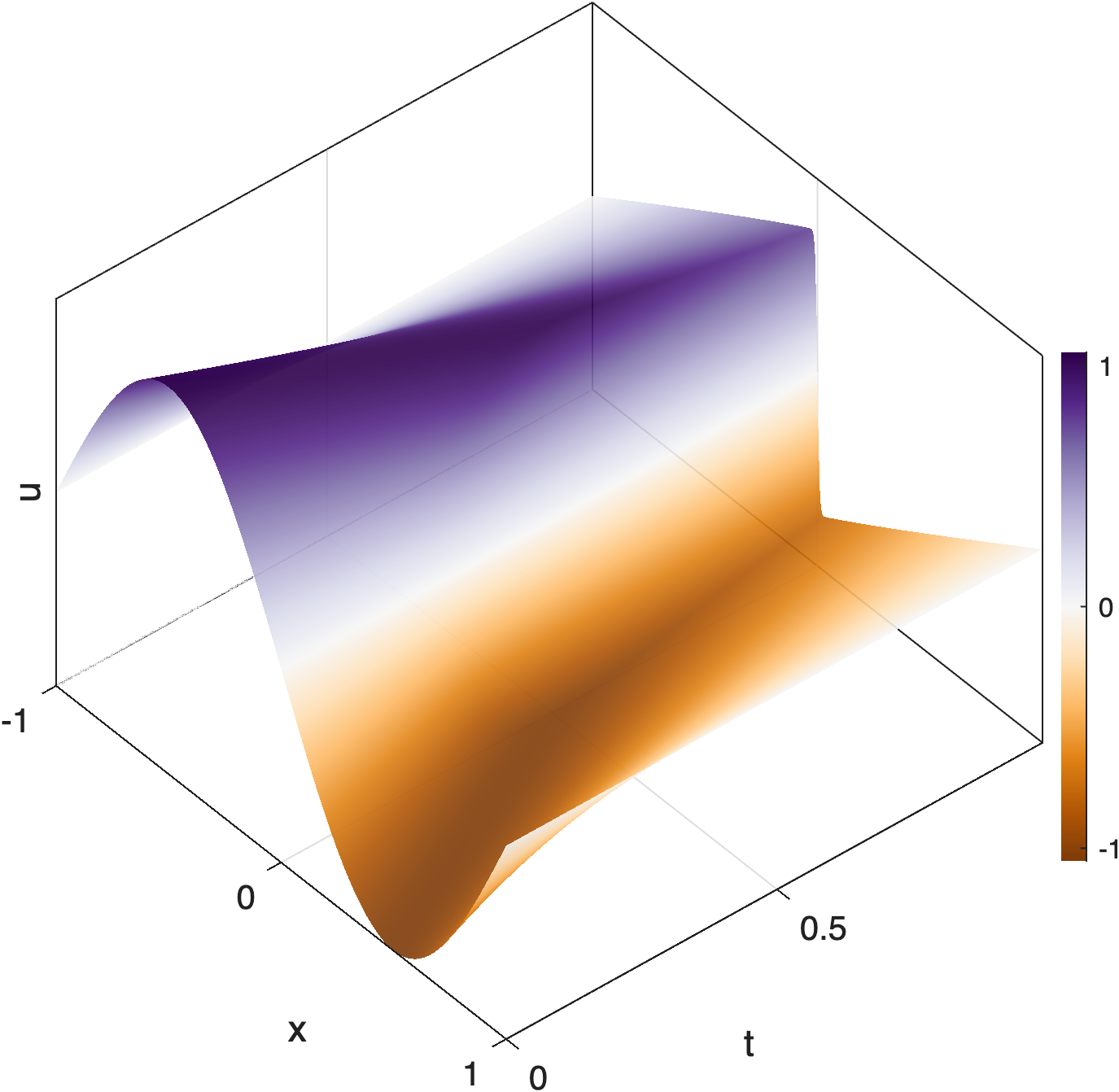}
        \caption{Burgers equation}
        \label{fig:test-burgers}
    \end{subfigure}
        \begin{subfigure}[b]{.32\linewidth}
        \centering
        \includegraphics[width=\linewidth]{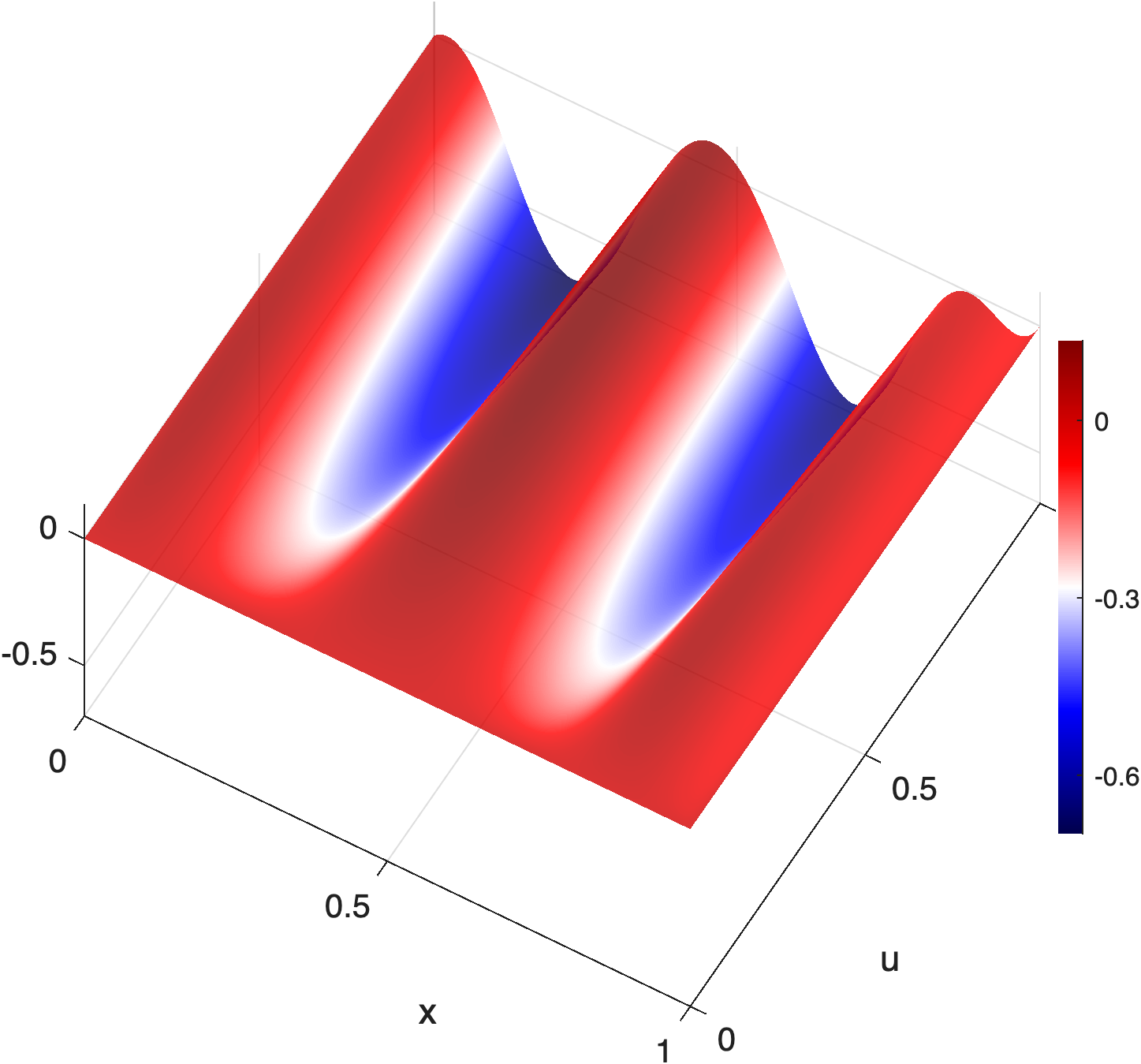}
        \caption{Diffusion-reaction equation}
        \label{fig:test-dre}
    \end{subfigure}
            \begin{subfigure}[b]{.32\linewidth}
        \centering
        \includegraphics[width=\linewidth]{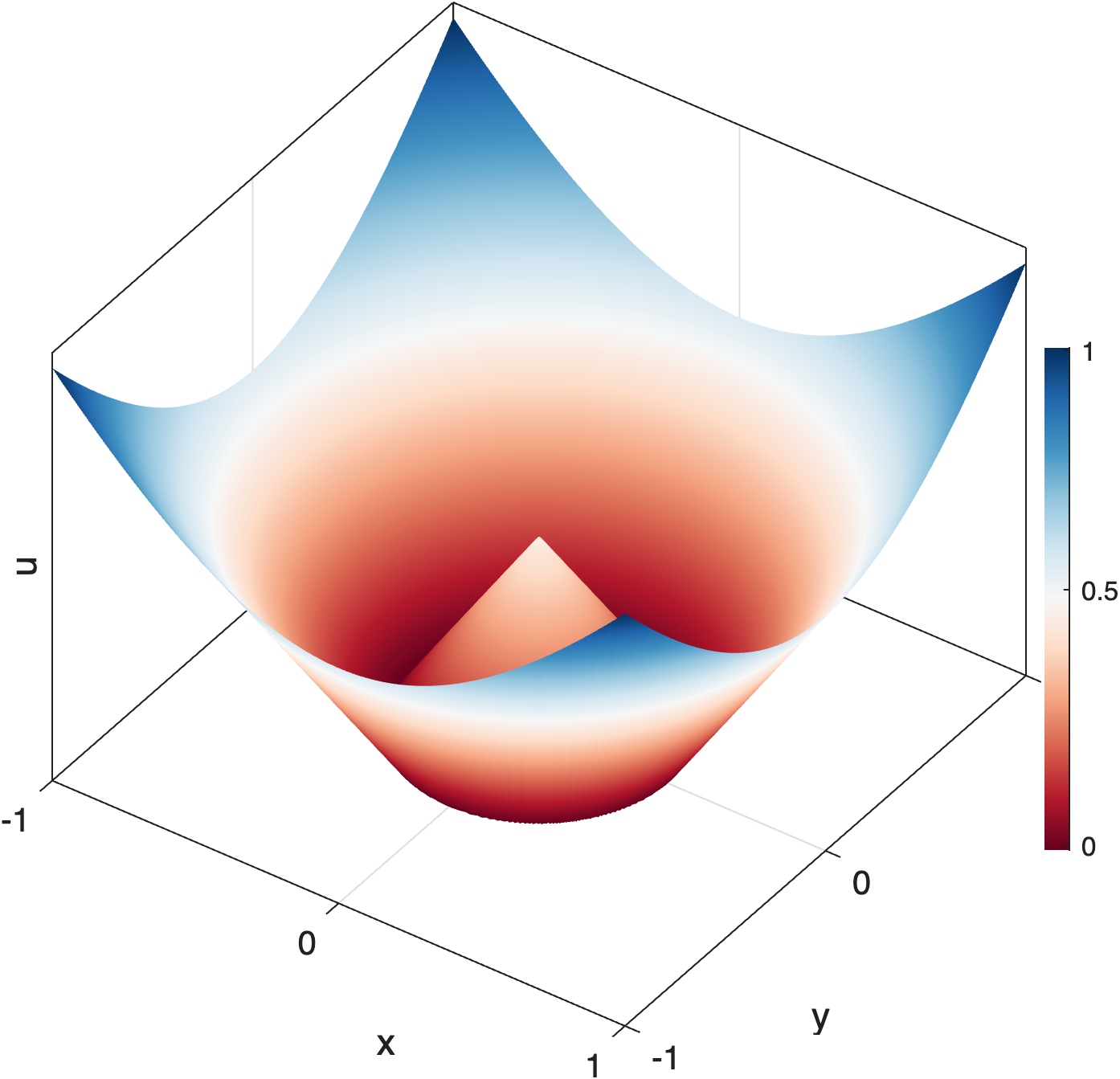}
        \caption{Eikonal equation}
        \label{fig:test-eikonal}
    \end{subfigure}
\caption{Test problems: (a) Burgers equation, (b) diffusion–reaction equation, and (c) Eikonal equation.}
    \label{fig:test-problem}
\end{figure}
\subsection{Burgers equation}
We first consider the one-dimensional Burgers’ equation benchmark investigated in~\cite{li2021fourier,wang2021learning}.  
The Burgers’ equation is a canonical nonlinear PDE that often develops sharp gradients and near-discontinuous features, thereby making it a challenging test case for neural operator models to capture accurately. The one dimensional Burgers equation with periodic boundary conditions (BC) and a given initial condition (IC) yields the following:
\begin{align}
\begin{aligned}
      &\frac{\partial u}{\partial t} 
    + u \frac{\partial u}{\partial x} 
    - \nu \frac{\partial^2 u}{\partial x^2} 
    = 0, 
    &&(x,t) \in (0,1) \times (0,1], \\[2mm]
    &\text{IC:  }u(x,0) = u_0(x), 
    &&x \in (0,1),   \\
    &\text{BC:  }u(0,t) = u(1,t), 
    &&\frac{\partial u}{\partial x}(0,t) = \frac{\partial u}{\partial x}(1,t). 
\end{aligned}
\end{align}
Following the setup in \cite{wang2021learning,li2021fourier}, we consider the time domain $t \in [0,1]$, with the viscosity parameter set to $\nu = 0.01$. The initial condition $u_0(x)$ is sampled from a Gaussian random field:
\begin{align}
    u_0(x) \sim \mathcal{N}\bigl(0,\, 25^2(-\Delta + 5^2 I)^{-4}\bigr),
\end{align}
which also enforces the periodic boundary conditions.

\paragraph{Training and testing data} To construct the training and testing datasets, we follow the setup in \cite{wang2021learning}. 
Specifically, we randomly sample $2000$ input functions from a Gaussian random field (GRF) 
$\mathcal{N}\big(0, 25^2(-\Delta + 5^2 I)^{-4}\big)$. 
Among these, a subset of $N = 1000$ realizations is selected as the training data. 
For each input sample $u$, the Burgers’ equation is solved using a spectral method with periodic boundary conditions, 
starting from $s(x,0) = u(x)$ for $x \in [0,1]$ and integrated to $t = 1$. 
Synthetic data are generated via the \texttt{Chebfun} package~\cite{driscoll2014chebfun} using a Fourier discretization and the ETDRK4 scheme~\cite{cox2002exponential} with a time step of $10^{-4}$.
Temporal snapshots of the solution are recorded every $\Delta t = 0.01$, resulting in $101$ snapshots in total. 
To assess model generalization, two different test datasets are considered: 
(i) a set of $100$ realizations drawn from the same GRF distribution as the training data, and 
(ii) an additional $400$ realizations independently sampled from the GRF. 
Each dataset is evaluated on a $100 \times 101$ spatio-temporal grid.

All models were trained for $150,000$ gradient descent iterations using the Adam optimizer with an exponentially decaying learning rate. The computations were performed on an NVIDIA GeForce RTX 5060 Ti GPU utilizing the Windows Subsystem for Linux (WSL) environment for improved compatibility and performance.

\paragraph{Neural network architectures}
Following \cite{wang2021learning}, the branch and trunk components are each equipped by a seven-layer fully connected neural network. 
Both networks employ $\tanh$ activation functions and comprise $100$ neurons per hidden layer. 
For consistency and fair comparisons, the same architecture is adopted for all models, including DeepONet, PI-DeepONet, and PAS-Net.

\paragraph{Error metric}
To quantitatively evaluate the predictive accuracy of different operator learning methods, we employ the time averaged $L^2$ error as the error metric. 
In particular: the $L^2$ error is defined as
\begin{align}
{\mathcal{E}}_{L^2} = \sqrt{\frac{1}{NM} \sum_{j=1}^{N} \sum_{k=1}^{M} 
\big| \tilde{u}(x_j, t_k) - \hat{u}(x_j, t_k) \big|^2},
\end{align}
where $\tilde{u}$ denotes the neural network solution and $\hat{u}$ represents the high-fidelity numerical solver result. 
To account for variability across different input conditions or parameter realizations, we further compute the \emph{mean} $L^2$ error by averaging ${\mathcal{E}}_{L^2}$ over an ensemble of test samples. 
Let $S$ denote the number of test realizations (ensembles), each corresponding to a distinct input function or parameter instance with supscript $(\cdot)^{(s)}$.
The ensemble-mean $L^2$ error is then defined as
\begin{align}
\overline{\mathcal{E}}_{L^2} 
= \frac{1}{S} \sum_{s=1}^{S} 
{\mathcal{E}}_{L^2}^{(s)}
= \frac{1}{S} \sum_{s=1}^{S} 
\sqrt{\frac{1}{NM} 
\sum_{j=1}^{N} \sum_{k=1}^{M} 
\big| \tilde{u}^{(s)}(x_j, t_k) - 
\hat{u}^{(s)}(x_j, t_k) \big|^2}.
\label{eq:mean-error}
\end{align}
The corresponding standard deviation (std) of the ensemble yields the folloing:
\begin{align}
\sigma_{\mathcal{E}_{L^2}}
= \sqrt{ \frac{1}{S-1} 
\sum_{s=1}^{S} 
\left( 
\mathcal{E}^{(s)}_{L^2}
- \overline{\mathcal{E}}_{L^2} 
\right)^2 }.
\label{eq:std-error}
\end{align}
Unless otherwise specified, all reported results correspond to the mean $L^2$ error $\overline{\mathcal{E}}_{L^2}$ computed over the test ensemble.

Table \ref{tab:l2-burgers} shows the mean $L^2$ errors and corresponding standard deviation of the one-dimensional Burgers equation for DeepONet, PI-DeepONet and PAS-Net with different $\gamma$ values.
Under the same number of training iterations, the plain vanilla DeepONet shows the largest errors which implies a slower convergence rate and a limited ability to capture the sharp gradients and nonlinear advection–diffusion behavior in the Burgers dynamics.
With the inclusion of the physics-based loss, PI-DeepONet significantly reduces both the mean and variability of the prediction errors. Further improvements are achieved by the proposed PAS-Net for all $\gamma$ settings that implies the adaptive-scale feature enhances convergence and enables more accurate representation of multi-scale features and localized shocks.
Specifically, the case $\gamma=-1$ achieves the lowest mean and standard deviation errors, corresponding to nearly a half reduction in prediction error relative to PI-DeepONet.
The PAS-Net performances are consistent with different sample sizes in the testing case, i.e., $S=100$ and $S=400$, which indicates that PAS-Net generalizes well with limited training data. 
This observation is further confirmed by Figure~\ref{fig:burgers-result} with benchmark solutions in Figure \ref{fig:burgers-bench}, which shows the spatiotemporal fields for three different initial conditions obtained by PI-DeepONet and PAS-Net under various $\gamma$ values. The DeepONet result is not shown for clarity, as its error magnitude is considerably larger than those of the other models, consistent with the values reported in Table~\ref{tab:l2-burgers}. In comparison, PAS-Net produces noticeably smaller error fields than PI-DeepONet for all three different initial conditions.


\begin{table}[htp!]
    \centering
    \caption{Mean relative $L^2$ errors and standard deviations of the Burgers equation obtained by DeepONet, PI-DeepONet, and PAS-Net for different values of $\gamma$.}
    \label{tab:l2-burgers}
    \vspace{2mm}
    \renewcommand{\arraystretch}{1.2}
    \begin{tabular}{lc|cc|cccc}
        \toprule
        & & & & \multicolumn{4}{c}{\textbf{PAS-Net (different $\gamma$ values)}} \\
        \cmidrule(lr){5-8}
        \textbf{Sample} & \textbf{Metric} & \textbf{DeepONet} & \textbf{PI-DeepONet} & $\boldsymbol{\gamma=-\tfrac{1}{3}}$ & $\boldsymbol{\gamma=-\tfrac{1}{2}}$ & $\boldsymbol{\gamma=-\tfrac{2}{3}}$ & $\boldsymbol{\gamma=-1}$ \\
        \midrule
        \multirow{2}{*}{$100$} & Mean $\overline{\mathcal{E}}_{L^2}$ & $2.38\times 10^0$ & $3.34\times 10^{-2}$ & $2.03\times 10^{-2}$ & $2.21\times 10^{-2}$ & $2.08\times 10^{-2}$ & $1.84\times 10^{-2}$ \\
        & Std $\mathcal{E}_{L^2}$& $1.56\times 10^0$ & $3.45\times 10^{-2}$ & $2.49\times 10^{-2}$ & $2.98\times 10^{-2}$ & $2.55\times 10^{-2}$ & $2.05\times 10^{-2}$ \\
        \midrule
        \multirow{2}{*}{$400$} & Mean $\overline{\mathcal{E}}_{L^2}$ & $2.41\times 10^0$ & $3.15\times 10^{-2}$ & $1.93\times 10^{-2}$ & $2.09\times 10^{-2}$ & $2.07\times 10^{-2}$ & $1.84\times 10^{-2}$ \\
        & Std $\mathcal{E}_{L^2}$ & $1.57\times 10^0$ & $2.81\times 10^{-2}$ & $1.51\times 10^{-2}$ & $1.63\times 10^{-2}$ & $1.75\times 10^{-2}$ & $1.32\times 10^{-2}$ \\
        \bottomrule
    \end{tabular}
\end{table}

\begin{figure}[H]
\includegraphics[width=\textwidth]{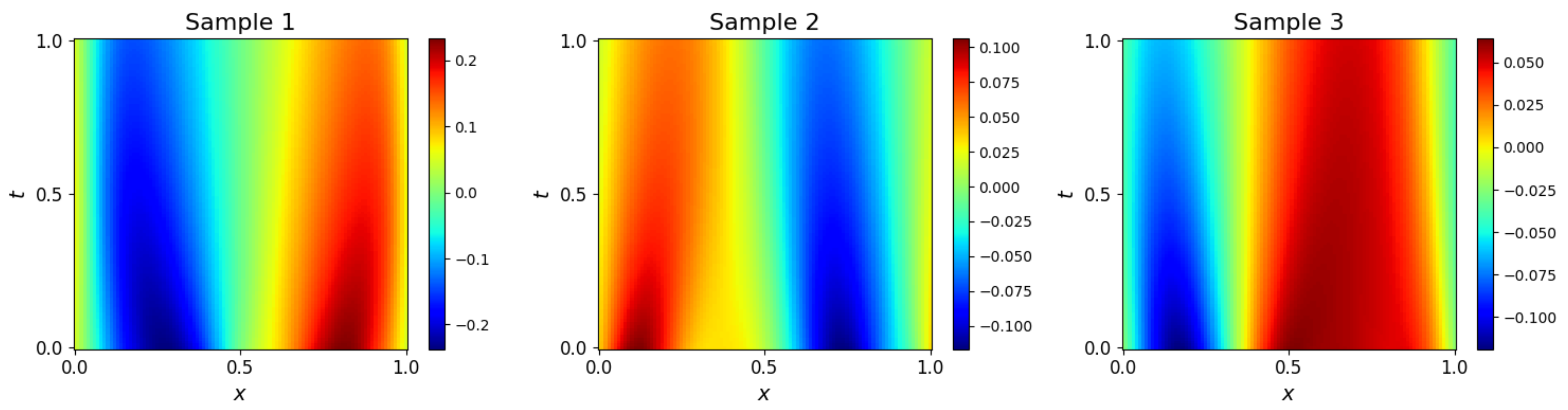}
    \caption{Benchmark solutions of three different samples used in Figure \ref{fig:burgers-result} for one-dimensional Burgers equation. }
    \label{fig:burgers-bench}
\end{figure}
\begin{figure}[H]
    \centering
    \begin{subfigure}[b]{\linewidth}
        \centering
        \includegraphics[width=\linewidth]{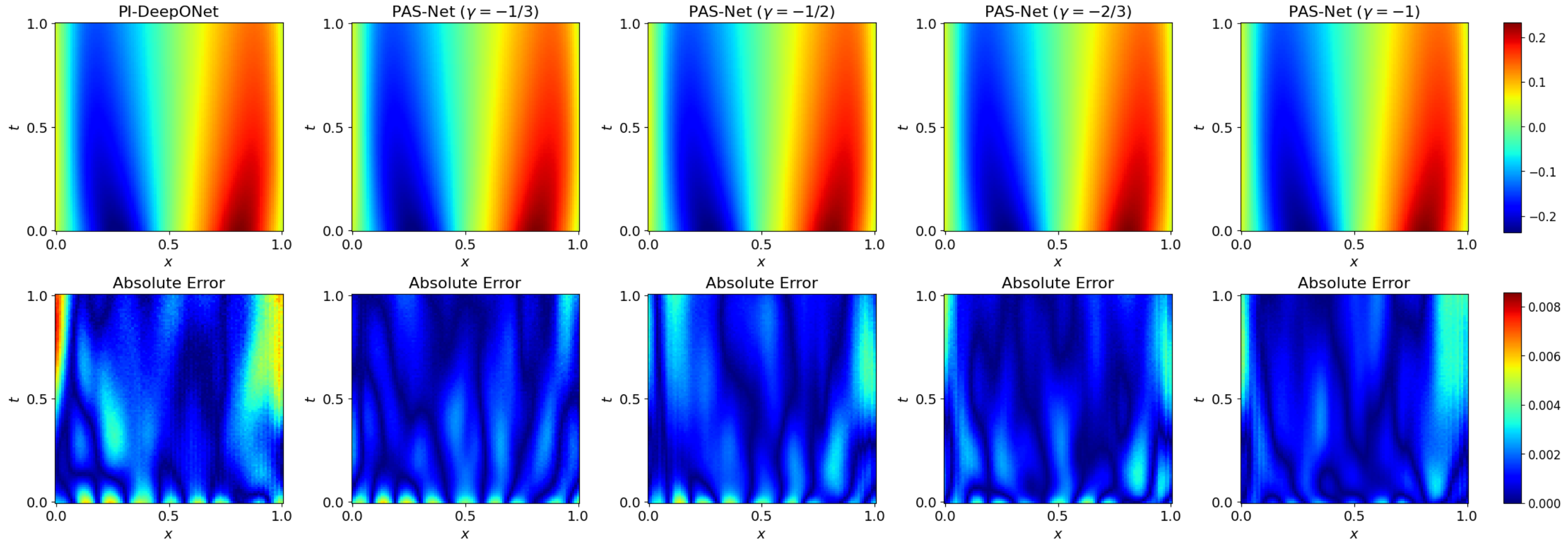}
        \caption{Sample 1}
        \label{fig:burgers-sample-1}
    \end{subfigure}
    \\
    \begin{subfigure}[b]{\linewidth}
        \centering
        \includegraphics[width=\linewidth]{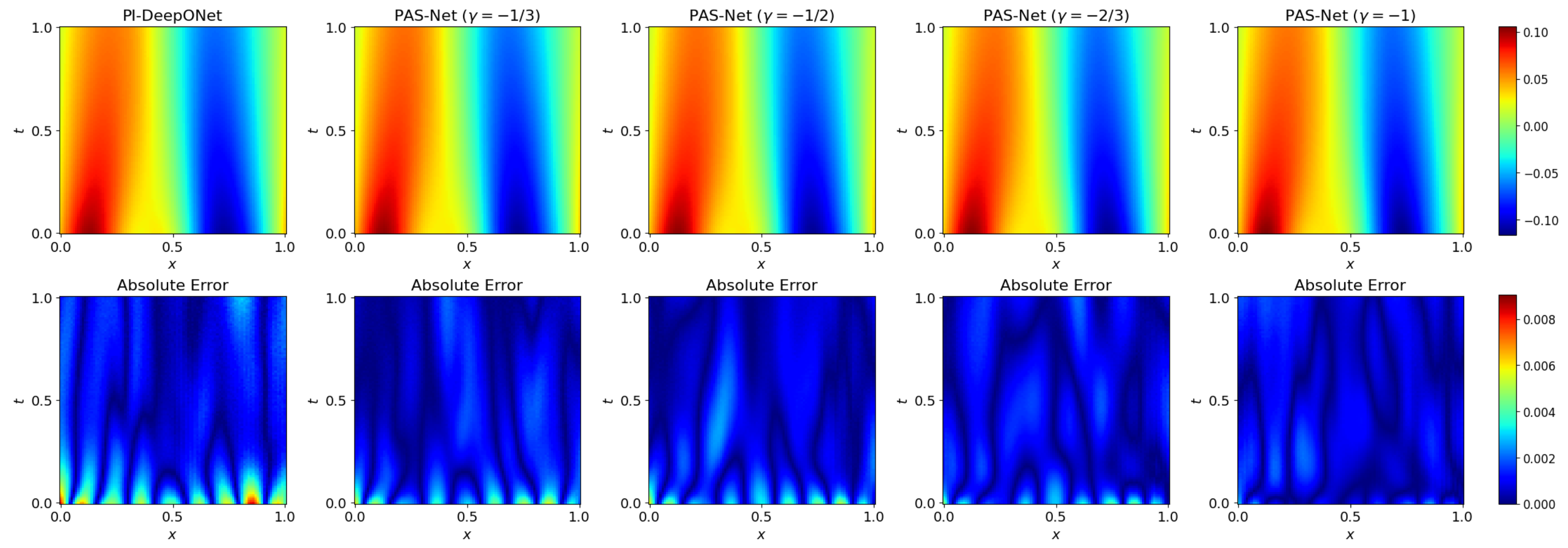}
        \caption{Sample 2}
        \label{fig:burgers-sample-2}
    \end{subfigure}  
        \\
    \begin{subfigure}[b]{\linewidth}
        \centering
        \includegraphics[width=\linewidth]{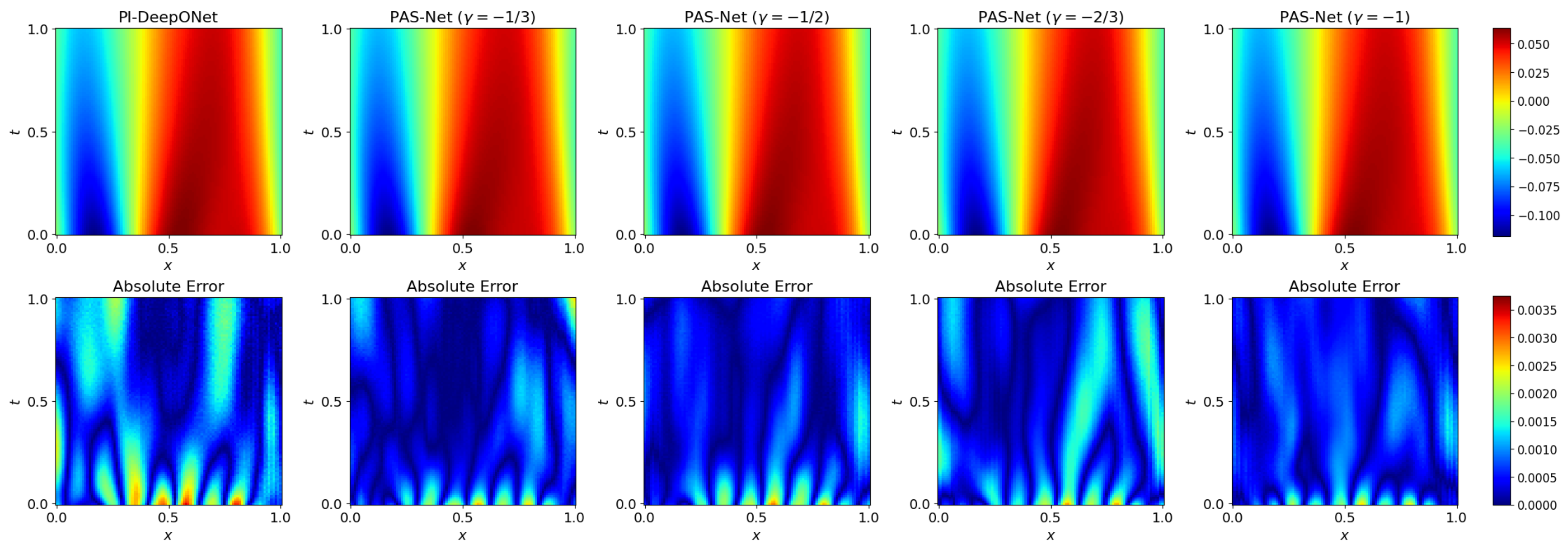}
        \caption{Sample 3}
        \label{fig:burgers-sample-3}
    \end{subfigure} 
    \caption{Spatiotemporal fields of one dimensional Burgers equation with different initial conditions for PI-DeepONet and PAS-Net with different $\gamma$ values. }
    \label{fig:burgers-result}
\end{figure}


\subsection{Diffusion-reaction equations}
Our next example for PAS-Net involves a nonlinear diffusion--reaction equation with Dirichlet boundary conditions and a spatial forcing term $f(x)$ \cite{wang2021learning}:
\begin{align}
\frac{\partial u}{\partial t} 
= \nu\frac{\partial^2 u}{\partial x^2} + k u^2 + f(x), 
\quad (x,t) \in (0,1)\times(0,1),    
\end{align}
where $\nu = 0.01$ is the diffusion coefficient and $k = 0.01$ is the reaction rate. 
We assume zero initial condition. 
The objective similar tois to learn the solution operator that maps source terms $f(x)$ to the corresponding PDE solutions $u(x,t)$. 
The training data consist of random realizations of $f(x)$ sampled from a Gaussian random field (GRF), as described in \cite{wang2021learning}.

\paragraph{Training and testing data} 
For each \( u^{(i)} \), the boundary point set \( \{(x_{u,j}^{(i)}, t_{u,j}^{(i)})\}_{j=1}^{P} \) is sampled from the unit square boundary \([0,1] \times [0,1]\) excluding \( t=1 \). The collocation points \( \{(x_{r,j}^{(i)}, t_{r,j}^{(i)})\}_{j=1}^{Q} \) have fixed spatial coordinates \( x_{r,j}^{(i)} = x_j \) with temporal coordinates \( t_{r,j}^{(i)} \) uniformly sampled over \([0,1]\). We use \( P=300 \) and \( Q=100 \) , generating 5,000 training input functions \( u(x) \) from a GRF with length scale \( l=0.2 \). The test dataset contains 100 input functions sampled from the same GRF, where corresponding diffusion-reaction systems are solved using a second-order implicit finite difference method on a \( 100 \times 100 \) equidistant grid. Each test solution is evaluated on this uniform grid.


\paragraph{Neural network architectures} 
The branch and trunk networks are each equipped by a five-layer fully connected neural network with 50 neurons per hidden layer.
We conduct 120,000 gradient descent iterations to train the models, focusing on minimizing the loss function. 

\paragraph{Error metric} We use the same error metrics \eqref{eq:mean-error}--\eqref{eq:std-error} as in the Burgers equation test case.

Table~\ref{tab:l2-diffusion} reports the mean and standard deviation of the relative $L^2$ errors for the one-dimensional diffusion–reaction equation obtained by DeepONet, PI-DeepONet, and PAS-Net with different parameters.
Similar to the Burgers case, the plain vanilla DeepONet still shows the highest errors within the limited number of training iterations which indicates that it struggles to accurately capture the spatial variability induced by the nonlinear reaction term.
Introducing the physics-informed loss again leads to a substantial reduction in both the mean and variability of the prediction errors in PI-DeepONet, which confirms that incorporating physical constraints improves generalization for reaction–diffusion processes.
Among all different parameter, PAS-Net consistently attains the lowest mean $L^2$ errors with different combinations of $(\epsilon, x_c, \gamma)$. In particular, the settings with $\epsilon=0.1$ and $\gamma=-0.6$ or $-0.7$ yield slightly better accuracy, showing that moderate localization around $x_c=0.7$ improves the learning of spatially varying reaction fronts.
These observations are further supported by Figure~\ref{fig:DR-result} with benchmark solutions in Figure \ref{fig:dre-bench}, which presents the predicted spatiotemporal fields for three representative initial conditions obtained by PI-DeepONet and PAS-Net. The DeepONet results are still omitted for clarity due to their comparatively larger error magnitude. 

It is worth noting that the improvement of PAS-Net compared with PI-DeepONet in this case is less pronounced than that in the Burgers equation. This difference arises because, in advection-dominated regimes, the steep gradients are continuously transported in space, making them more challenging to track and represent accurately. While the adaptive-scale feature enhances local resolution, its time-invariant spatial scaling may not fully capture the dynamic movement of sharp fronts. Nevertheless, PAS-Net still consistently achieve lower errors than PI-DeepONet and remains an efficient framework for learning such advection–diffusion processes.

\begin{figure}[H]
\includegraphics[width=\textwidth]{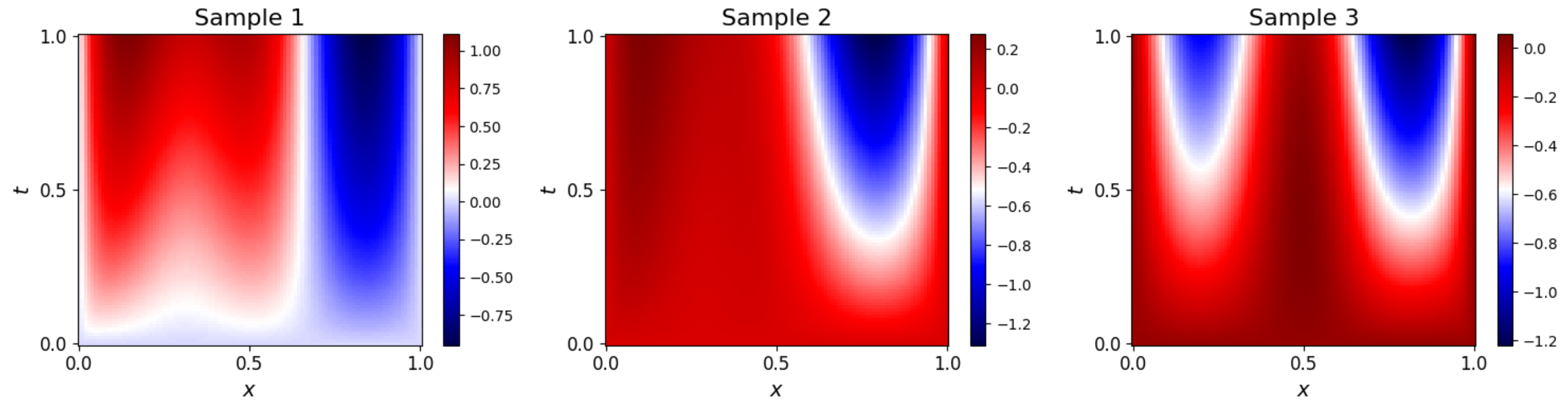}
    \caption{Benchmark solutions of three different samples used in Figure \ref{fig:DR-result} for one-dimensional diffusion-reaction equation. }
    \label{fig:dre-bench}
\end{figure}
\begin{table}[H]
    \centering
    \caption{Mean $L^2$ errors of the Diffusion-reaction equation obtained by DeepONet, PI-DeepONet, and PAS-Net for different parameters.}
    \label{tab:l2-diffusion}
    \vspace{2mm}
    \renewcommand{\arraystretch}{1.2}
    \begin{tabular}{l|cc|cccc}
        \toprule
        & \multicolumn{2}{c|}{} & \multicolumn{4}{c}{\textbf{PAS-Net}} \\
        \cmidrule(lr){2-3} \cmidrule(lr){4-7}
        & \textbf{DeepONet} & \textbf{PI-DeepONet} 
        &  {\footnotesize$\boldsymbol{\epsilon=0.1,  x_c=0.7}$} & {\footnotesize$\boldsymbol{\epsilon=0.1,  x_c=0.75}$} & {\footnotesize$\boldsymbol{\epsilon=0.5,  x_c=0.75}$} & {\footnotesize$\boldsymbol{\epsilon=0.1, x_c=0.7}$} \\
        & & & {\footnotesize$\boldsymbol{\gamma=-0.6}$} & {\footnotesize$\boldsymbol{\gamma=-0.5}$} & {\footnotesize$\boldsymbol{\gamma=-2.0}$} & {\footnotesize$\boldsymbol{\gamma=-0.7}$} \\
        \midrule
        Mean $\mathcal{E}_{L^2}$ & $8.00\times 10^{-3}$ & $4.65\times 10^{-3}$ & $4.26\times 10^{-3}$ & $4.52\times 10^{-3}$ & $4.46\times 10^{-3}$ & $4.43\times 10^{-3}$ \\
        Std $\sigma_{\mathcal{E}_{L^2}}$ & $4.08\times 10^{-3}$ & $2.31\times 10^{-3}$ & $2.03\times 10^{-3}$ & $2.18\times 10^{-3}$ & $2.14\times 10^{-3}$ & $2.03\times 10^{-3}$ \\
        \bottomrule
    \end{tabular}
\end{table}
\begin{figure}[H]
    \centering
    \begin{subfigure}[b]{\linewidth}
        \centering
        \includegraphics[width=\linewidth]{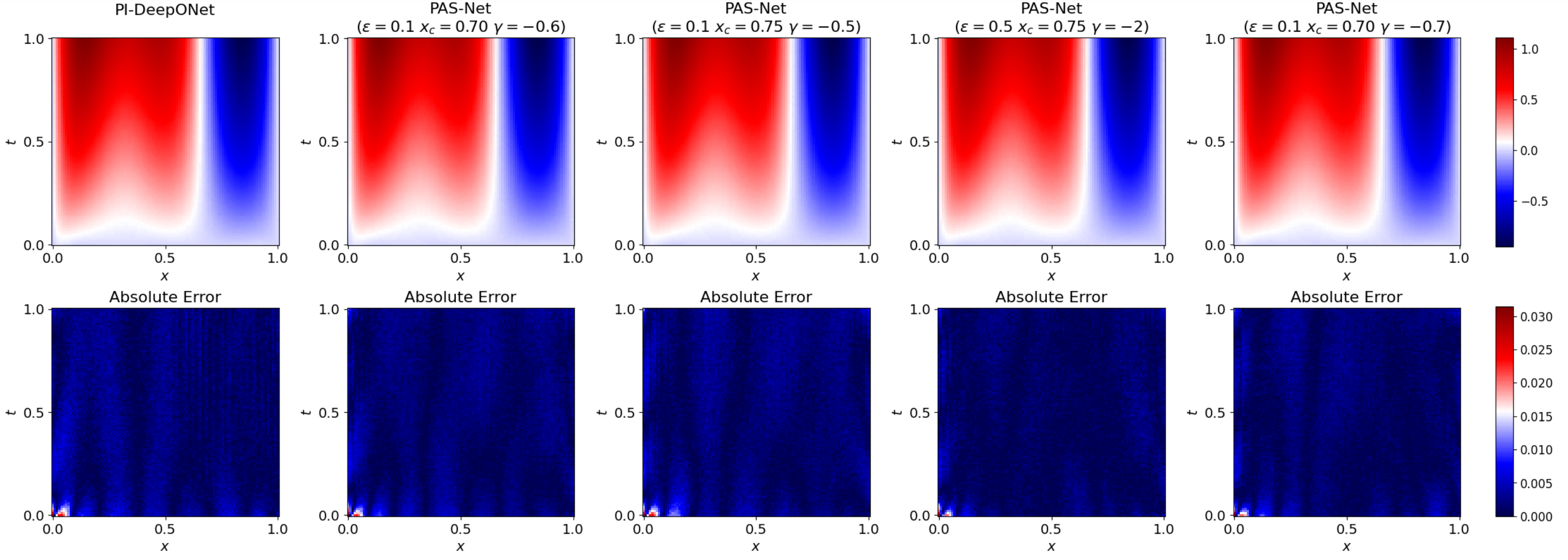}
        \caption{Sample 1}
        \label{fig:DR-sample-1}
    \end{subfigure}
    \\
    \begin{subfigure}[b]{\linewidth}
        \centering
        \includegraphics[width=\linewidth]{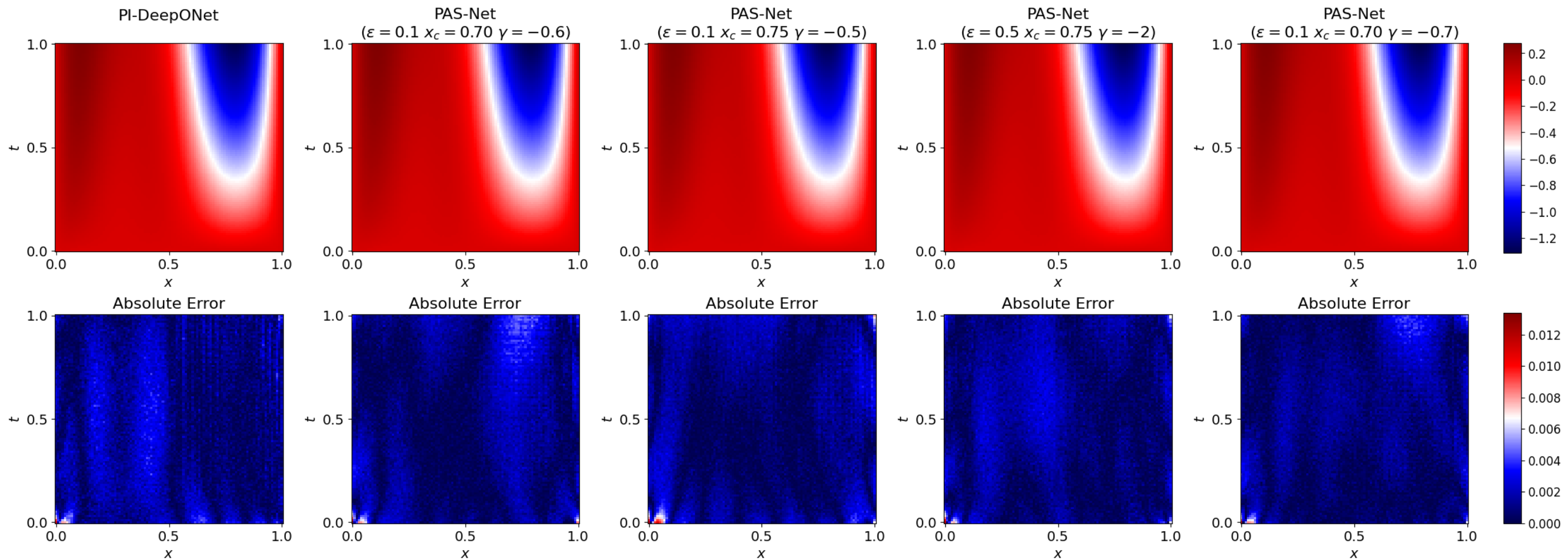}
        \caption{Sample 2}
        \label{fig:DR-sample-2}
    \end{subfigure}  
        \\
    \begin{subfigure}[b]{\linewidth}
        \centering
        \includegraphics[width=\linewidth]{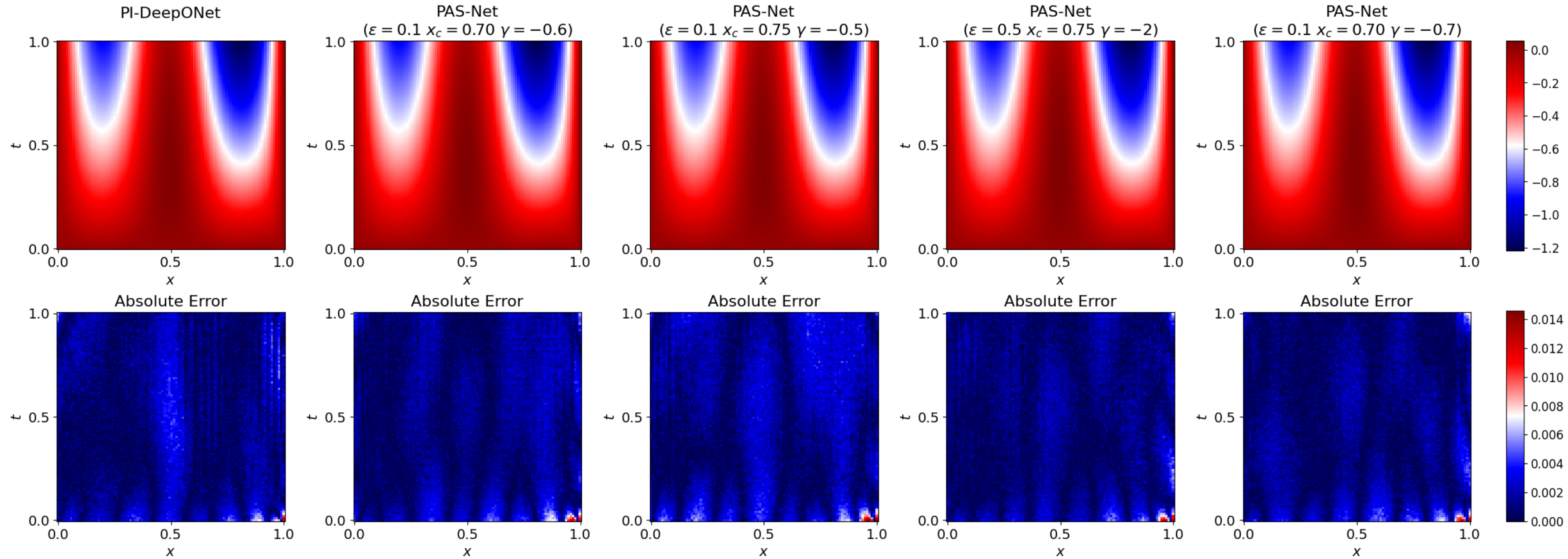}
        \caption{Sample 3}
        \label{fig:DR-sample-3}
    \end{subfigure} 
    \caption{Spatiotemporal fields of one dimensional diffusion-reaction equation with different initial conditions for PI-DeepONet and PAS-Net with different parameters. }
    \label{fig:DR-result}
\end{figure}

\subsection{Eikonal equation}

The final example illustrates the ability of PAS-Net to approximate operators in the application from nonlinear Hamilton–Jacobi–type equations. Specifically, we consider the two-dimensional eikonal equation \cite{wang2021learning}:
\begin{align}
    \|\nabla u(\mathbf{x})\|_2 &= 1, \notag\\
    u(\mathbf{x}) &= 0, \quad \mathbf{x} \in \partial \Omega, \notag
\end{align}
where $\mathbf{x} = (x,y) \in \mathbb{R}^2$ represents the spatial coordinates, and $\Omega = [-2, 2] \times [-2, 2]$ is an open domain with a smooth boundary $\partial \Omega$. 
The solution $u(\mathbf{x})$ corresponds to a signed distance function that measures the shortest distance from any point $\mathbf{x} \in \Omega$ to the boundary $\partial \Omega$, which can be expressed as follows:
\begin{align}
u(\mathbf{x}) =
\begin{cases}
d(\mathbf{x}, \partial \Omega), & \mathbf{x} \in \Omega,\\
-d(\mathbf{x}, \partial \Omega), & \mathbf{x} \in \Omega^c,
\end{cases} \notag
\end{align}
where the distance function ``$d(\mathbf{x}, \partial \Omega)$'' is defined:
\begin{align}
d(\mathbf{x}, \partial \Omega) = \inf_{\mathbf{y} \in \partial \Omega} d(\mathbf{x}, \mathbf{y}). \notag
\end{align}
In this example, different operator learning models are trained to learn the mapping from a closed boundary curve $\Gamma$ to its corresponding signed distance function $u(\mathbf{x})$, which satisfies the eikonal equation above. 
\paragraph{Training and testing data} 
For the construction of the training dataset, 1000 circles are randomly selected with all their radii sampled from a uniform probability distribution. For each input circle \(\Gamma^{(i)}\) with radius \( r^{(i)}\), the corresponding point set \(\{(x_j^{(i)}, y_j^{(i)})\}_{j=1}^Q\) is defined as \(\{(r^{(i)} \cos \theta_j, r^{(i)} \sin \theta_j)\}_{j=1}^Q\), where the angles \(\{\theta_j\}_{j=1}^Q\) follow a uniform distribution within the interval \([0, 2\pi]\). We  set  \( m = 100 \) ,  \( Q = 1000 \) and the computational domain as \( D = [-1,1] \times [-1,1] \) . 

\paragraph{Neural network architectures} 
The branch network is equipped by a five-layer fully connected neural network, and the trunk network six, each hidden layer in both networks contains 50 neurons. We conduct 40,000 gradient descent iterations to train the models, focusing on minimizing the loss function. 


\paragraph{Error metric} We use the same error metrics \eqref{eq:mean-error}--\eqref{eq:std-error} as in the Burgers equation test case.


\paragraph{Choice of parameters in PAS-Net}
To test the influence of both the scaling exponent $\gamma$ and the spatial localization of the adaptive centers $(x_c, y_c)$ in PAS-Net in two-dimensional domain, we consider the following four different regimes for the proposed PAS-Net with a fixed scaling parameter $\epsilon=0.5$.
For simplicity, these regimes are referred to as Regime~1, Regime~2, Regime~3, and Regime~4 as in Table \ref{tab:pasnet-params}. 
\begin{table}[H]
    \centering
    \caption{Parameter regimes used in PAS-Net for the two-dimensional Eikonal equation.}
    \label{tab:pasnet-params}
    \vspace{2mm}
    \renewcommand{\arraystretch}{1.2}
    \begin{tabular}{c|c|c|c}
        \toprule
        \textbf{Regime} & \textbf{ $\gamma$} & { $\epsilon$} & \textbf{$(x_c, y_c)$} \\
        \midrule
        1 & $-1$    & $0.5$ & $(0.50,\,0.50)$ \\
        2 & $-1$    & $0.5$ & $(0.75,\,0.75)$ \\
        3 & $-2.10$ & $0.5$ & $(0.50,\,-0.50)$ \\
        4 & $-0.33$ & $0.5$ & $(0.75,\,-0.75)$ \\
        \bottomrule
    \end{tabular}
\end{table}

Table~\ref{tab:l2-eikonal} summarizes the mean and standard deviation of the relative $L^2$ errors for the two-dimensional Eikonal equation obtained using DeepONet, PI-DeepONet, and PAS-Net under the four parameter regimes defined in Table~\ref{tab:pasnet-params}.
Same as preivous two test cases, the plain vanilla DeepONet has the largest errors and PI-DeepONet reduces both the mean and standard deviation errors by approximately $17\%$ and $53\%$ compared with DeepONet, respectively. 
PAS-Net further improves the overall accuracy for all parameter regimes. In particular, Regime~4 ($\gamma=-0.33$, $(x_c,y_c)=(0.75,-0.75)$) has the smallest mean and standard deviation errors, reducing by about $55\%$ and the standard deviation by nearly $60\%$ compared to PI-DeepONet. This improvement implies that adaptive scaling could effectively capture the circular wavefronts and maintain stability in regions of high curvature. The small variation among the four PAS-Net configurations also demonstrates the robustness of PAS-Net for all different parameter choices. This observation is also consistent with Figure~\ref{fig:Eikonal-result} shows three representative sample cases of the spatial fields for the two-dimensional Eikonal equation, with the corresponding benchmark solutions shown in Figure~\ref{fig:Eikonal-bench}. The DeepONet result is omitted for clarity and to maintain a consistent comparison between PI-DeepONet and PAS-Net. For all cases, PAS-Net displays systematically smaller spatial errors than PI-DeepONet, particularly near curved and intersecting wavefronts.



\begin{figure}[H]
\includegraphics[width=\textwidth]{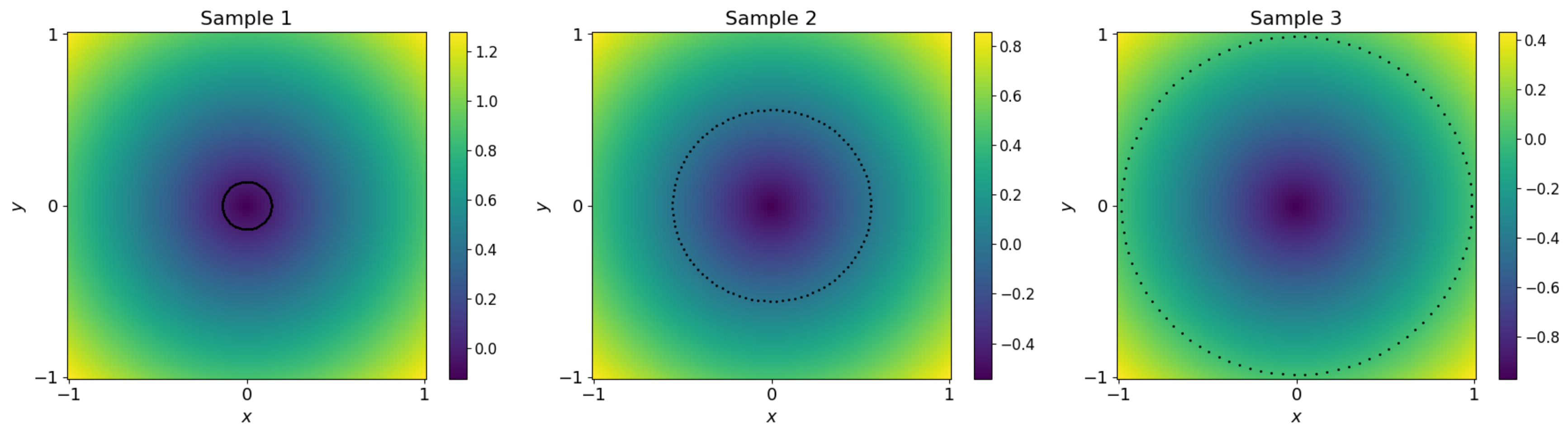}
    \caption{Benchmark solutions of three different samples used in Figure \ref{fig:Eikonal-result} for two-dimensional Eikonal equation. }
    \label{fig:Eikonal-bench}
\end{figure}
\begin{table}[H]
    \centering
    \caption{Mean $L^2$ errors of the Eikonal equation obtained by DeepONet, PI-DeepONet, and PAS-Net for different parameters.}
    \label{tab:l2-eikonal}
    \vspace{2mm}
    \renewcommand{\arraystretch}{1.2}
    \begin{tabular}{l|cc|cccc}
        \toprule
        & & & \multicolumn{4}{c}{\textbf{PAS-Net}} \\
        \cmidrule(lr){2-3} \cmidrule(lr){4-7}
        & \textbf{DeepONet}& \textbf{PI-DeepONet}& \textbf{Regime1} & \textbf{Regime2} & \textbf{Regime3} & \textbf{Regime4} \\
        \midrule
        Mean $\mathcal{E}_{L^2}$ & $1.00\times 10^{0}$ & $7.54\times 10^{-3}$ & $5.36\times 10^{-3}$ & $4.19\times 10^{-3}$ & $5.28\times 10^{-3}$ & $3.36\times 10^{-3}$ \\
        Std $\sigma_{\mathcal{E}_{L^2}}$ & $4.03\times 10^{-5}$ & $2.00\times 10^{-3}$ & $1.43\times 10^{-3}$ & $1.27\times 10^{-3}$ & $1.48\times 10^{-3}$ & $8.35\times 10^{-4}$ \\
        \bottomrule
    \end{tabular}
\end{table}

\begin{figure}[H]
    \centering
    \begin{subfigure}[b]{\linewidth}
        \centering
        \includegraphics[width=\linewidth]{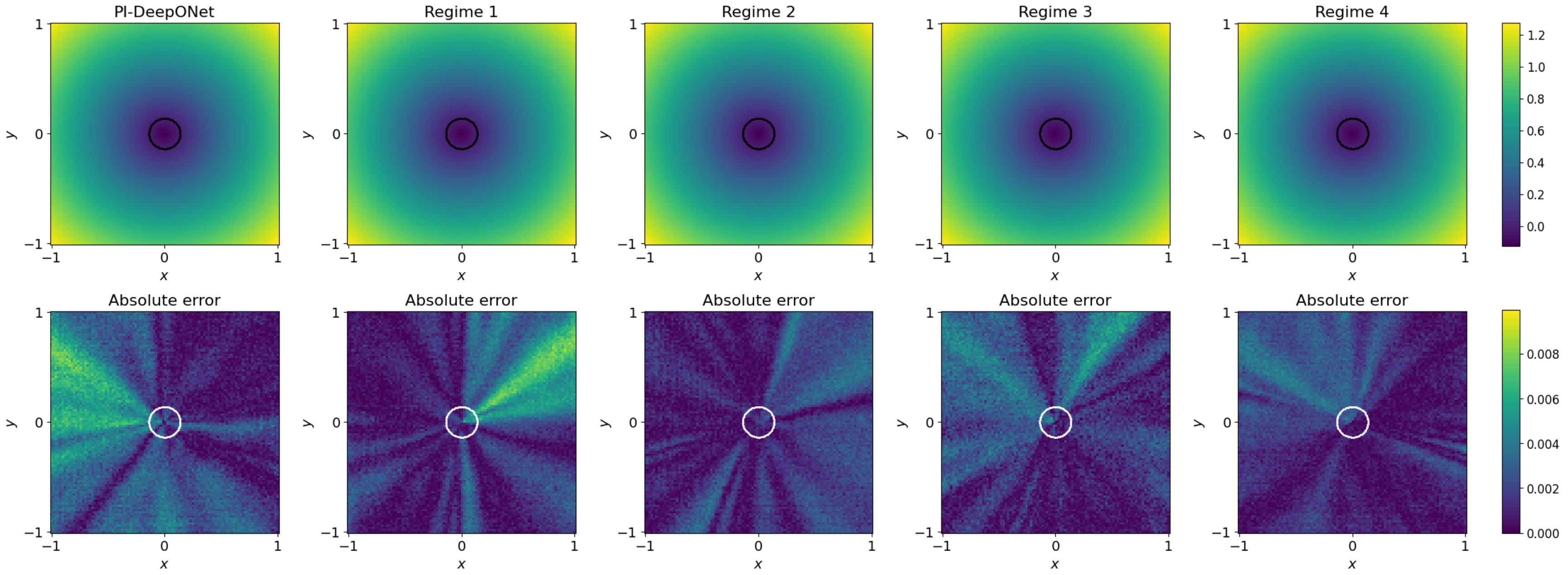}
        \caption{Sample 1}
        \label{fig:Eikonal-sample-1}
    \end{subfigure}
    \\
    \begin{subfigure}[b]{\linewidth}
        \centering
        \includegraphics[width=\linewidth]{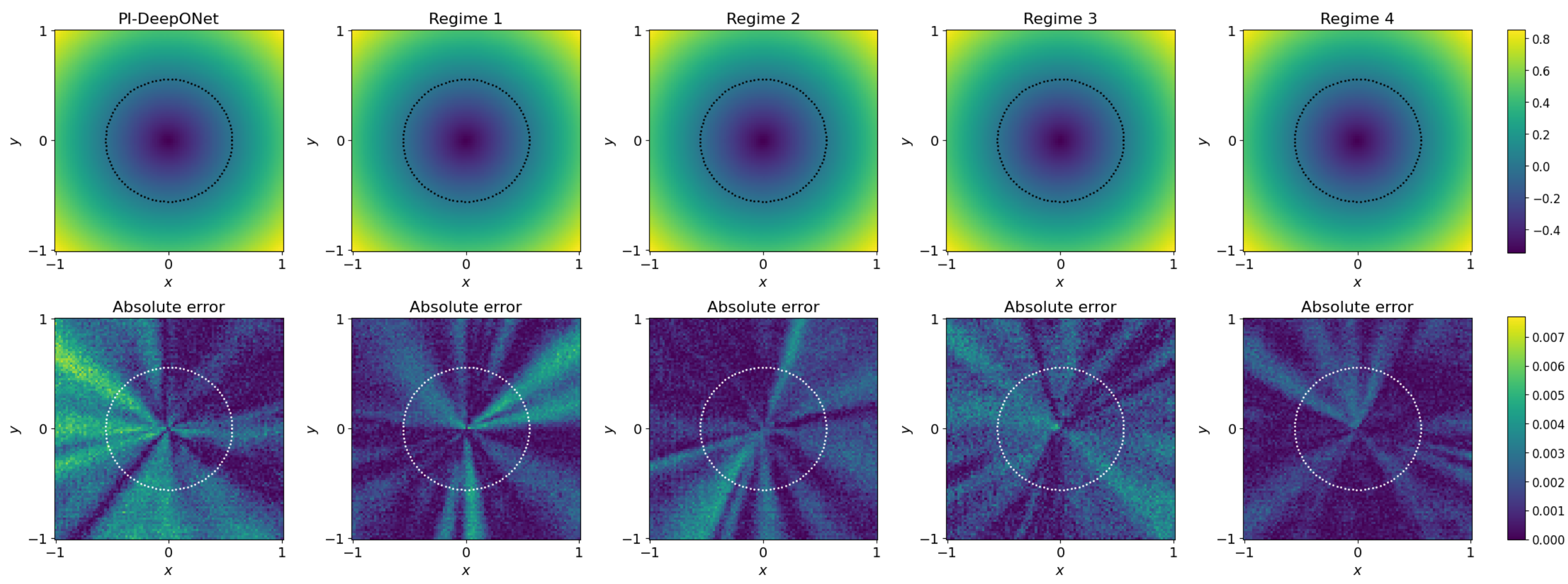}
        \caption{Sample 2}
        \label{fig:Eikonal-sample-2}
    \end{subfigure}  
        \\
    \begin{subfigure}[b]{\linewidth}
        \centering
        \includegraphics[width=\linewidth]{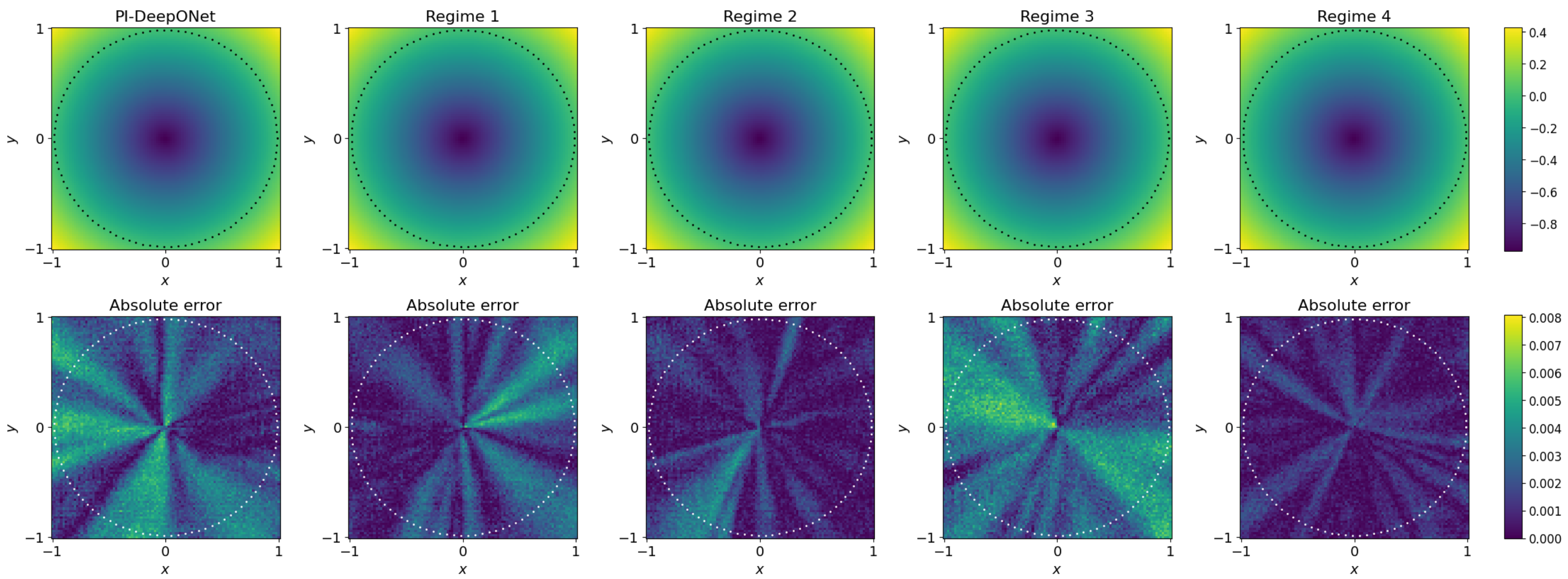}
        \caption{Sample 3}
        \label{fig:Eikonal-sample-3}
    \end{subfigure} 
    \caption{Spatial fields of Eikonal equation with different initial conditions for PI-DeepONet and PAS-Net with different parameter regimes. }
    \label{fig:Eikonal-result}
\end{figure}

\section{Conclusion and Future Work \label{sec:conclusion}}
We proposed the \emph{Physics-Informed Adaptive-Scale Deep Operator Network} (PAS-Net), a framework designed to improve operator learning for nonlinear and singularly perturbed PDEs with localized and multiscale features.  
By augmenting the trunk input of PI-DeepONet with locally rescaled coordinate embeddings, PAS-Net introduces an intrinsic multiscale representation that enhances expressivity and training stability.  
Theoretical analysis based on the Neural Tangent Kernel (NTK) shows that this adaptive embedding improves spectral conditioning and accelerates convergence.  
In this paper, numerical tests on Burgers, advection, and diffusion–reaction equations show that PAS-Net achieves substantially more accurate than the standard DeepONet and PI-DeepONet models, especially in regimes characterized by steep gradients or stiff dynamics. 

Future work will proceed along several directions.  
First, PAS-Net can be extended to physics-informed inverse modeling, where it may be integrated with advanced operator-learning frameworks, e.g., Fourier Neural Operators, Neural U-Nets, or Transformer-based architectures, to leverage complementary spectral and attention-based representations for recovering unknown parameters or source terms in complex PDE systems \cite{lin2021multi,kharazmi2021identifiability,lu2025evolutionary,xu2024modeling}.
Second, a more comprehensive numerical and theoretical analysis in the continuous-operator setting will be pursued to better understand convergence, generalization, and approximation properties under adaptive scaling.  
Third, PAS-Net will be extended to more challenging multiscale systems, including incompressible Navier–Stokes and Boussinesq equations \cite{mou2021data,zhang2022convergence,zhang2021fast,zhang2021effects}, where nonlinear coupling and turbulence provide more demanding tests for physics-informed operator learning.  
Finally, we plan to explore applications of PAS-Net to \emph{interface problems} in PDEs \cite{guo2021error,adjerid2020error,guo2021construct}, where discontinuities and jump conditions across heterogeneous media require improved local adaptivity and robust scale handling.

\section*{Acknowledgements}
Y.Z is grateful to acknowledge the support of the National Natural Science Foundation of China (NSFC) under Grants No. 12401562 , No. 12571459, and No. 12241103.

\section*{Data and Code Availability}
All datasets and implementation codes supporting the findings of this study are available in the corresponding GitHub repository:  
\url{https://github.com/zhuxuewen-rgb/PAS-NET}.  
The repository includes the trained models, numerical solvers, and scripts used to reproduce all results presented in this paper.


\bibliographystyle{elsarticle-num}
\bibliography{MsNN,nn,intro-future}

\end{document}